\documentclass[a4paper,11pt]{article}

\usepackage{amsmath}
\usepackage{fancyhdr}
\usepackage{amssymb}
\usepackage{amsthm}
\usepackage{graphicx}
\usepackage{varioref}
\usepackage{verbatim} 
\usepackage{multicol}
\usepackage{enumerate}
\usepackage[normalem]{ulem}
\usepackage{caption}
\usepackage{float}
\usepackage{subcaption}
\usepackage[a4paper,margin=2.5cm]{geometry}
\usepackage[linesnumbered,ruled,vlined]{algorithm2e}

\usepackage{tikz}
\usetikzlibrary{patterns}
\usetikzlibrary{arrows}
\usetikzlibrary{decorations.pathreplacing,calligraphy}

\usepackage{mathrsfs}
\usepackage{mathtools}

\usepackage{url}
\usepackage{xspace}
\usepackage{thm-restate}

\usepackage{microtype}

\usepackage{tikz}
\usetikzlibrary{calc}

\usepackage{epstopdf}

\usepackage[colorlinks = true,pagebackref=true]{hyperref}
\hypersetup{
  pdftitle = {Quantum Path-Edge Sampling}
}
\usepackage{xcolor}
\definecolor{darkred}  {rgb}{0.5,0,0}
\definecolor{darkblue} {rgb}{0,0,0.5}
\definecolor{darkgreen}{rgb}{0,0.5,0}
\hypersetup{
  urlcolor   = blue,         
  linkcolor  = darkblue,     
  citecolor  = darkgreen,    
  filecolor  = darkred       
}

\usepackage[capitalise,nameinlink]{cleveref}
\crefname{lemma}{Lemma}{Lemmas}
\crefname{proposition}{Proposition}{Propositions}
\crefname{definition}{Definition}{Definitions}
\crefname{theorem}{Theorem}{Theorems}
\crefname{conjecture}{Conjecture}{Conjectures}
\crefname{corollary}{Corollary}{Corollaries}
\crefname{section}{Section}{Sections}
\crefname{appendix}{Appendix}{Appendices}
\crefname{figure}{Figure}{Figures}
\crefname{equation}{Eq.}{Eqs.}
\crefname{table}{Table}{Tables}
\crefname{claim}{Claim}{Claims}
\crefname{item}{Item}{Items}

\newtheorem{theorem}{Theorem}
\newtheorem{lemma}[theorem]{Lemma}

\newtheorem{definition}[theorem]{Definition}
\newtheorem{corollary}[theorem]{Corollary}

\newtheorem*{conjecture*}{Conjecture}

\theoremstyle{definition}

\newcommand{\Lal}[1]{{\Lambda^{\!{#1}}}}
\newcommand{\U}[3]{{U({#1},{#2},{#3})}}
\newcommand{\w}[2]{{w_+({#1},{#2})}}
\newcommand{\Aa}[1]{{\tilde{A}^{#1}}}
\newcommand{\PtHx}{\Pi_{x}}
\newcommand{\PHx}{\Pi_{{H}(x)}}

\newcommand{\tp}{\ket{\psi_{x,0}}}
\newcommand{\tm}{\ket{\psi_{x,+}}}
\newcommand{\tpb}{\bra{\psi_{x,0}}}

\newcommand{\tpmbk}{\braket{\psi_{x,0}}{\psi_{x,+}}}

\newcommand{\q}{q_{G(x),s,t}}

\def\abs#1{\left| #1 \right|}
\def\norm#1{\left\| #1 \right\|}
\newcommand{\eps}{\varepsilon}

\usepackage{authblk}

\newcommand{\ket}[1]{|#1\rangle}
\newcommand{\bra}[1]{\langle#1|}
\newcommand{\proj}[1]{|#1\rangle\!\langle#1|}

\DeclareMathAlphabet{\matheu}{U}{eus}{m}{n}

\newcommand{\sop}[1]{{\mathcal #1}}

\newcommand{\braket}[2]{\langle{#1}|{#2}\rangle}
\newcommand{\ketbra}[2]{|{#1}\rangle\!\langle{#2}|}

\newcommand{\ceil}[1]{\left\lceil{#1}\right\rceil}

\begin{document}

\title{Quantum Algorithm for Path-Edge Sampling}
\author[1]{Stacey Jeffery}
\author[2]{Shelby Kimmel}
\affil[1]{CWI \& QuSoft}
\author[1]{Alvaro Piedrafita}
\affil[2]{Middlebury College, \texttt{skimmel@middlebury.edu}}

\date{}

\maketitle

\begin{abstract} 
We present a quantum algorithm for sampling an edge on a path between two nodes $s$ and $t$ in an undirected graph given as an adjacency matrix, and show that this can be done in query complexity that is asymptotically the same, up to log factors, as the query complexity of detecting a path between $s$ and $t$. We use this path sampling algorithm as a subroutine for $st$-path finding and $st$-cut-set finding algorithms in some specific cases. Our main technical contribution is an algorithm for generating a quantum state that is proportional to the positive witness vector of a span program. \end{abstract}


\section{Introduction} 

Finding and detecting paths between two vertices in a graph are important
related problems, both in and of themselves, and as subroutines in other
applications, but there is still much to understand in this area. While
classically these problems seem to be equivalent, an intriguing question is
whether the same holds for quantum algorithms: there are cases where a
quantum algorithm can \emph{detect} a path between $s$ and $t$ in
significantly less time than any known quantum algorithm takes to 
\emph{find} such a path. In particular, path finding on a glued trees graph is one of
Aaronson's top ten open problems in query complexity 
\cite{childs2003exponential,aaronson2021open}, 
as the best known quantum algorithms that find an $st$-path in such graphs have 
exponentially worse running time than the best quantum algorithms for detecting 
one, and understanding how these problems are related could improve our understanding of why quantum computers achieve dramatic speedups for certain problems. As an example of more immediate practical interest: path finding in supersingular isogeny graphs is one approach to attacking cryptosystems based on supersingular isogenies~\cite{charlesCryptHash2007,defeoSupersingular2014}, but currently the best known attack of this form still takes exponential time~\cite{taniClaw2009} (see also~\cite{galbraithSupersingular2018}). 

In this paper, we consider the quantum query complexity of a somewhat
intermediate problem: finding an edge on an $st$-path in an undirected graph.\footnote{In this paper,
we use \emph{path} to refer to a self-avoiding path, meaning a path with no
repeated vertices.} In the classical case, it seems hard to imagine how one
could find an edge on an $st$-path without first finding an $st$-path, but we
show that in the quantum case, one can sample an $st$-path edge with similar
resources to what is needed to detect the existence of an $st$-path. In some
cases, this can be done with significantly fewer queries than the best
previously known path-finding algorithms. We show this ability to sample an
edge on a path has some useful applications, including to sabotaging
networks (finding $st$-cut sets) and to finding paths in certain graphs faster than
existing path finding algorithms.

Previously, D\"urr, Heiligman, H{\o}yer and Mhalla 
\cite{durr2006quantum} described an algorithm for connectivity in the adjacency
matrix model that uses $O(n^{3/2})$ queries for an $n$-vertex graph. Their algorithm works by keeping track of known connected components, and
then uses a quantum search to look for any edge that connects any two
components previously not known to be connected. While the authors use this
algorithm to decide connectivity, we note that after $O(n^{3/2})$ queries,
the algorithm will produce (with high probability) a list of the connected
components of the graph, as well as a set of edges for each component that is
a witness to that component's connectivity (a spanning tree). This data can
then be used to find a path from $s$ to $t$, if $s$ and $t$ are in the same
component. This algorithm uses $O(\log n)$ qubits and $O(n\log n)$ classical
bits, and applies to both directed and undirected graphs. 

However, the algorithm of D\"urr et~al.~does not take advantage of any
structure in the graph. This is in contrast to an undirected path \textit
{detection} quantum algorithm of Belovs and Reichardt~\cite{belovsSpanProgramsQuantum2012}, further analyzed and refined in~\cite{jarretQuantumAlgorithmsConnectivity2018,anderson2020leverage}, which, for example, can detect a path between vertices $s$
and $t$ with $\widetilde{O}(\sqrt{L}n)$ adjacency matrix queries when there is an 
$st$-path of length $L$, and even better in the case of multiple short paths, or in the case of certain promises when there is no path.
In fact, there are even sufficiently structured promises on the input for which this algorithm performs superpolynomially better than the best possible classical algorithm~\cite{jefferyQuantumAlgorithmsGraph2017}. 
While this path detection algorithm runs faster than $O(n^{3/2})$ in many cases, the algorithm does not output any
information about the $st$-path -- it simply determines whether a
path exists.

Our contribution is an algorithm that reproduces the query complexity of the Belovs-Reichardt undirected
path detection algorithm, even for structured inputs -- for example, our algorithm uses $\widetilde{O}(\sqrt
{L}n)$ queries when there is a path of length $L$ -- but now returns \emph
{some} information about edges on an $st$-path: namely, a path edge.\footnote{As we hinted at with our statement of advantages for the Belovs-Reichardt algorithm in the case of shorter and/or multiple paths, the Belovs-Reichardt algorithm for $st$-path detection actually has a complexity that depends on the structure of the graph in a more subtle way, replacing $L$ with an upper bound on the \emph{effective resistance} between $s$ and $t$, which is \emph{at most} the length of the shortest path between $s$ and $t$. This more subtle analysis also applies to our edge finding algorithm.} Specifically, our algorithm
outputs an $st$-path edge sampled with probability that depends on the
optimal $st$-flow between $s$ and $t$. This is how electrons would flow in an
electrical network if edges in the graph were replaced by wires with
resistors and a battery were connected between $s$ and $t$. For intuition, an
edge is more likely to be sampled if it is on \emph{more} or \emph
{shorter} paths. Thus, in the case of a single path between $s$ and $t$, our
algorithm samples each edge in the path with equal probability (up to some error in total variation distance). When there
are disjoint paths of different lengths, our algorithm is more likely to
sample an edge on a short path than a long path -- the probability of
sampling from a particular path of length $\ell$ is proportional to $1/\ell$.
(This means, unfortunately, that if there are many long paths, we might still
be more likely to sample an edge on some long path than an edge on a short
path). We prove that finding an $st$-path edge classically requires $\Omega
(n^2)$ queries in the worst case, even if promised that there is a path of
length $L$, as long as $L\geq 3$.

With the ability to quickly find edges on short paths, we can create an 
improved algorithm for \textit{finding} $st$-paths in undirected graphs
with a unique, short $st$-path. Given an adjacency matrix for an $n$-vertex graph,
if there is a unique $st$-path, whose (possibly unknown) length is $L$, we can find all of the edges
in the path in $\widetilde{O}(L^{1+o(1)}n)$ expected queries. When $L=o(\sqrt{n})$, this is an
improvement over the D\"urr et~al.~algorithm. In the general case that there
is more than one $st$-path, we prove that we can find all edges
in a single path in $\widetilde{O}(L^{3/2}n)$ queries when $L$ is the (possibly unknown) length of the 
\textit{longest} path (although our approach
in this case does not use the edge sampling algorithm as a subroutine).
 When $L=o(n^{1/3})$, this is an
improvement over the D\"urr et~al.~algorithm. 

We additionally use our sampling algorithm to find $st$-cut sets, in the case
that $s$ and $t$ are each part of
a highly connected component, and there are only a few edges connecting
those components. Because these few connecting edges are bottlenecks in the flow,
there will be a lot of flow over those connecting edges, and so a high probability
of sampling them, and hence finding an $st$-cut set. We describe a particular family
of $n$-vertex graphs were we can find such a cut set in $\widetilde{O}(n)$ queries, where
any classical algorithm would require $\Omega(n^2)$ queries.

Our edge sampling algorithm is a special case of a new span-program-based
algorithm (\cref{sec:witness-gen}) for generating quantum states called \emph{span program witness states} 
(or simply \emph{witness states}).
One of the
key elements of the analysis of span program algorithms for deciding Boolean functions \cite{reichardtReflectionsQuantumQuery2011} is the positive
witness (see \cref{def:posWitness}), which is a vector that witnesses that
the function evaluates a particular input to $1$. While in the usual span
program algorithm, the output on input $x$ is $f(x)$, in our case, we output
a quantum state proportional to the positive witness for input $x$. In the
case of the Belovs-Reichardt span program for $st$-connectivity~\cite
{belovsSpanProgramsQuantum2012}, a positive witness is a linear
combination of edges that are on paths between $s$ and $t$, where the
amplitudes depend on the optimal $st$-flow (see \cref
{def:st-flow}). Generating and then measuring such a state allows us to
sample $st$-path edges.

Our results more generally hold for the case where the input $x$ defines a subgraph $G(x)$ of some arbitrary graph $G$, that is not necessarily a complete graph. Although we do not attempt to analyze time complexity in this work, we suspect that our query algorithms on graphs are also time efficient when there
is an efficient way to perform a quantum walk on the underlying
graph $G$, as in~\cite{jefferyQuantumAlgorithmsGraph2017}. For example, when $G$ is the $n$-vertex complete graph (i.e. the oracle allows you
to query elements of the full adjacency matrix for a $n$-vertex graph, as we have been assuming throughout this introduction), there
is an efficient way to do this walk, and so in this case the time complexity
of our algorithms is likely the same as the query complexity, up to log factors.

\subsection{Future Directions}

A natural future direction is to try to use our edge finding technique for
path finding in more general settings than the ones we consider. One
surprising aspect of our algorithm is that it does not necessarily find edges in 
the order in which they appear in the path, and instead often finds edges 
in the middle of a path with high probability. 
The form of our algorithm thus seems to circumvent a recent lower bound
on path-finding in glued trees graphs that applies to algorithms that 
always maintain a path from the starting node
 to any vertex in the algorithm's state~\cite{childs2022quantum}.
 However, one
 reason to be pessimistic for this particular application is that in the glued trees graph, \emph{all} edges
 connected to the starting vertex
 are in some $st$-path. Still, we are hopeful that for some graphs, 
finding an
edge in the middle of some $st$-path opens up the possibility of new divide-and-conquer approaches for path finding. 

We are only able to take advantage of the fact that we sample edges 
according to the
optimal $st$-flow for very specific graphs, like those with a single path, or 
with bottleneck flows, but we hope that this edge sampling distribution will prove useful
in additional applications. In recent independent
work, Apers and Piddock \cite{apers2022elfs} develop a similar edge sampling
algorithm in the adjacency list model, which they use to analyze connections
between electric flows and quantum walks, and they prove that walks that
proceed via their edge sampling algorithm need only logarithmically many
rounds before they have a high probability of reaching a target vertex, on
trees. We believe that such edge sampling methods will likely find further applications. 

We have only applied our span program witness state generation algorithm to the span program for
path detection. Span program algorithms exist for a wide range of graph
problems, from bipartiteness \cite{cadeTimeSpaceEfficient2018} and cycle
detection \cite
{cadeTimeSpaceEfficient2018,delorenzoApplicationsQuantumAlgorithm2019a}, to
triangle \cite{carette2020extended} and other subgraph detection \cite
{lee2011learning}, to other combinatorial search problems \cite
{belovs2012learning,beigiQuantumSpeedupBased2019}. Perhaps the span program witness states for
these problems would be useful for certain
applications. Beyond span program algorithms, dual adversary algorithms (which are equivalent to span programs for decision problems, but generalize to state conversion problems \cite{leeQuantumQueryComplexity2011}) and
multidimensional quantum walks \cite
{jeffery2022multidimensional,jeffery2022quantum} all have a similar notion of
witnesses in their design and analysis. Similar techniques might yield
witness generation algorithms for these more general algorithm design
paradigms.

We suspect our path finding algorithms are not optimal, as for graphs with
longest paths of length $\Omega(n^{1/3})$, our algorithms do not outperform
D\"urr et~al.'s algorithm. We wonder whether it is possible to find paths
using $o(n^{3/2})$ queries whenever the longest path has length $o(n),$ or to
prove that this is not possible, perhaps by expanding on techniques for lower
bounding path-finding on welded trees \cite{childs2022quantum}.

Finally, all of our algorithms apply only to undirected graphs, while the
algorithm of \cite{durr2006quantum} applies equally well to directed 
or undirected graphs. While there are span program  algorithms
for problems on directed graphs (see e.g. \cite{beigiQuantumSpeedupBased2019}), they do not
exhibit the same speedups with short or many paths that the undirected
span program algorithms possess. It would be interesting to better
understand whether there are ways to obtain similar improvements in
query complexity for directed graphs.

\paragraph{Organization.}  In \cref{sec:witness-gen} we present our main
 technical result: an algorithm for generating a state proportional to a span
 program witness for $x$. In \cref{sec:applications}, we show how to apply
 this to finding a path edge (\cref{sec:edge-finding}), and give an example
 of a particular family of graphs in which the classical complexity of
 finding a path edge is quadratically worse than our quantum algorithm (\cref
 {thm:class_edge_finding}). In \cref{sec:sabotage}, we show how our edge
 finding algorithm can be applied to efficiently find an $st$-cut set in a particular family
 of graphs, and in \cref{sec:path-finding} we show how it can be applied to
 find an $st$-path in $\widetilde{O}(nL^{1+o(1)})$ queries when there is a \emph
 {unique} $st$-path of length $L$ (\cref{thm:singlePath}); and also give an algorithm for finding an $st$-path in general
 graphs in $\widetilde{O}(nL^{3/2})$ queries when $L$ is the length of
 the \emph{longest} $st$-path (\cref{thm:generalPathFinder}).

\section{Preliminaries}

We first introduce some basic notation. We let $\|\cdot\|$ denote the $l_2$
norm, $[m]\coloneqq\{1,2,3,\dots,m\}$, and let $\sop L(H,V)$ denote the set of linear operators from
the vector space $H$ to the vector space $V.$

\subsection{Span Programs}

Span programs are a linear algebraic model of computation, 
introduced in \cite{KW93}, 
that have proven extremely useful for analyzing query 
\cite{reichardtReflectionsQuantumQuery2011,reichardtSpanProgramsAre2014}, 
space \cite{jeffery2022span}, and time complexity 
\cite{belovsSpanProgramsQuantum2012,cornelissen2020span,beigi2022time} in 
quantum algorithms. We follow Ref.~\cite{itoApproximateSpanPrograms2019} closely 
in our definitions.

\begin{definition} [Span Program] For a finite set $R$, a span program on $R^m$ is a tuple 
$\mathcal{P}=(H,\mathcal{V},\ket{\tau},A)$ where
\begin{enumerate}
\item $H$ is a direct sum of finite-dimensional inner product 
spaces: $H=H_1\oplus H_2\cdots H_m\oplus H_\textrm{true}\oplus H_\textrm{false},$
and for $j\in [m]$ and $a\in R$, we have $H_{j,a}\subseteq H_j$, 
such that $\sum_{a\in R}H_{j,a}=H_j$;
\item $\mathcal{V}$ is a vector space;
\item $\ket{\tau}\in \mathcal{V}$ is a target vector; and
\item $A\in\sop L(H,\mathcal{V})$.
\end{enumerate} 
Given a string $x\in R^m$, we use $H(x)$ to denote the
 subspace $H_{1,x_1}\oplus\cdots \oplus H_{m,x_m}\oplus H_\textrm{true}$, and
 we denote by $\PHx$ the orthogonal projector onto the space $H(x)$.

\label{def:SP}
\end{definition}

An important concept in the analysis of span programs and quantum query complexity is that of \emph{witnesses}:
\begin{definition} [Positive Witness] Given a span program 
$\mathcal{P}=(H,\mathcal{V},\ket{\tau},A)$ on
 $ R^m$ and $x\in R^m$, $\ket{w}\in H(x)$ is a \emph{positive witness} for $x$
 in $\mathcal{P}$ if $A\ket{w}=\ket{\tau}$. If a positive witness exists for $x$, 
 we define
 the \emph{witness size} of $x$ in $\mathcal{P}$ as 
\begin{equation}
w_+(x)=\w{\mathcal{P}}{x}\coloneqq\min\left\{\|\ket{w}\|^2:\ket{w}\in H(x) \textrm{ and } A\ket{w}=\ket{\tau} \right\}.
\end{equation}
We say that $\ket{w}\in H(x)$ is the optimal positive witness 
for $x$ if $\|\ket{w}\|^2=\w{\mathcal{P}}{x}$ and $A\ket{w}=\ket{\tau}$.
\label{def:posWitness}
\end{definition}

Our main algorithm produces a normalized version of this unique optimal
positive witness, $\ket{w}/\|\ket{w}\|$. (To see that the optimal positive
witness is unique, for contradiction assume that the optimal positive witness
is not unique -- then a linear combination of two optimal positive witnesses
produces a witness with smaller witness size than either.) 

A span program $\mathcal{P}$ encodes a function $f:X\rightarrow\{0,1\}$ in the
following way. We say $f(x)=1$ if $x$ has a positive witness, and $f(x)=0$ if
$x$ does not have a positive witness. We say such a $\mathcal{P}$ decides the
function $f$.

We will also need the concept of an approximate negative witness.
\begin{definition} [Negative Error, Approximate Negative Witness]
Given a span program $\mathcal{P}=(H,\mathcal{V},\ket{\tau},A)$ on $ R^m$ and 
$x\in R^m$, we define the negative error of $x$ in $\mathcal{P}$ as
\begin{equation}
e_-(x,\mathcal{P})\coloneqq\min\left\{\|\bra{\widetilde{\omega}} A\PHx\|^2:\bra{\widetilde{\omega}}\in{\cal L}({\cal V},\mathbb{R}), \braket{\widetilde{\omega}}{\tau}=1 \right\}.
\end{equation}
 Note that $e_-(x,\mathcal{P})=0$ if and only if $\mathcal{P}$ decides a function
 $f$ with $f(x)=0$. Any $\bra{\widetilde{\omega}}$ such that 
 $\|\bra{\widetilde{\omega}} A\PHx\|^2=e_-(x,\mathcal{P})$ is
 called an approximate negative witness for $x$ in $P$. We define the
 approximate negative witness size of $x$ as:
\begin{equation}
\widetilde{w}_-(x,\mathcal{P})\coloneqq\min\left\{\|\bra{\widetilde{\omega}} A\|^2:\bra{\widetilde{\omega}}\in{\cal L}({\cal V},\mathbb{R}),\braket{\widetilde{\omega}}{\tau}=1,\|\bra{\widetilde{\omega}} A\PHx\|^2=e_-(x,\mathcal{P}) \right\}.
\end{equation}
We call an approximate negative witness $\bra{\widetilde{\omega}}$ that also minimizes 
$\|\bra{\widetilde{\omega}} A\|^2$ an optimal approximate negative witness.
\label{def:app_negWit}
\end{definition}

We use the following notation for maximum positive and approximate negative witness sizes:
\begin{align}
W_+(\mathcal{P},f)=W_+\coloneqq\max_{x\in f^{-1}(1)}\w{\mathcal{P}}{x},\qquad 
\widetilde{W}_-(\mathcal{P},f)=\widetilde{W}_-\coloneqq\max_{x\in f^{-1}(1)}\widetilde{w}_-(x,\mathcal{P}).
\end{align}
Note that we are restricting to 1-inputs of $f$. That is because our witness generation algorithm will assume that $x$ is a 1-input, unlike previous span-program-based algorithms that \emph{decide} $f$.

\subsection{Quantum Query Algorithms}\label{sec:prelim_quantum}

The algorithms we develop are query algorithms, where we can access a
unitary oracle $O_x$ for some $x\in X\subseteq R^m$ such that $O_x$ acts on
the space $\mathbb{C}^m\otimes \mathbb{C}^q$ as 
$O_x\ket{i}\ket{a}=\ket{i}\ket{x_i+a\mod q}$. where $q=|R|$, $x_i$ is the value of the $i^\textrm{th}$ element of $x$
and $\ket{i}\in \mathbb{C}^m$ and $\ket{a}\in \mathbb{C}^q$ are standard
basis states.

The query complexity of an algorithm is the number of times
$O_x$ must be used, in the worst case over $x\in X$. In our case, we will
also consider the expected query complexity on input $x$, 
which is the average number of
times $O_x$ must be used when given a particular input $x$, where the
randomness is due to random events in the course of the algorithm.

\subsection{Graph Theory and Connection to Span Programs} \label{sec:graph-theory}

Let $G=(V,E)$ be an undirected graph.\footnote{Our results easily extend to
multigraphs, see \cite{jarretQuantumAlgorithmsConnectivity2018}, but for
simplicity, we will not consider multigraphs here.}  We will particularly
consider graphs with specially labeled vertices $s,t\in V$, such that there
is a path from $s$ to $t$ in $G.$ Let 
$\overrightarrow{E}=\{(u,v):\{u,v\}\in E\}$; that is $\overrightarrow{E}$ is the set of directed edges
corresponding to the edges of $G.$ Given a graph $G=(V,E)$, for $u\in V$, we
denote by $G^-_u$ the subgraph of $G$ on the vertices 
$V\setminus\{u\}$, and with overloading of of notation for $S\subseteq E$, we denote by
$G^-_S$ the subgraph of $G$ with edges $S$ removed. (It will be clear from
context whether we are removing edges or vertices from the graph.) 

On a graph $G$ with $s$ and $t$ connected we will consider a \emph{unit $st$-flow}, which is a linear combination of cycles and $st$-paths, formally defined as a function on $\overrightarrow{E}$ with the following properties. 
\begin{definition}[Unit $st$-flow]\label{def:st-flow}
Let $G=(V,E)$ be an undirected graph with $s,t\in V(G)$, and $s$ and $t$ connected.
Then a \emph{unit $st$-flow} on $G$ is a function
$\theta:\overrightarrow{E}\rightarrow\mathbb{R}$ such that:
\begin{enumerate}
\item For all $(u,v)\in \overrightarrow{E}$, $\theta(u,v )=-\theta(v,u)$;
\item $\sum_{v:(s,v)\in \overrightarrow{E}}\theta(s,v)=\sum_{v:(v,t)\in \overrightarrow{E}}\theta(v,t)=1$; and 
\item for all $u\in V\setminus\{s,t\}$, $\sum_{v:(u,v)\in \overrightarrow{E}}\theta(u,v)=0$. 
\end{enumerate}
\end{definition}

\begin{definition}[Unit Flow Energy]\label{def:unitFlowEnergy}
Given a graph $G=(V,E)$ and a unit $st$-flow $\theta$ on $G$, the 
\emph{unit flow energy} of $\theta$ is 
$J(\theta)=\frac{1}{2}\sum_{e\in {\overrightarrow{E}}}\theta(e)^2.$

\end{definition}

\begin{definition}[Effective resistance] Let $G=(V,E)$ be a graph with $s,t\in V$.
If $s$ and $t$ are connected in $G$, the \emph{effective resistance} of
$G$ between $s$ and $t$ is $R_{s,t}(G) = \min_{\theta}
J(\theta)$, where $\theta$ runs over all unit $st$-unit flows of
$G$. If $s$ and $t$ are not connected in $G$, $R_{s,t}(G)=\infty.$
\label{def:effRes}
\end{definition}

\paragraph{Interpretation of the optimal flow} The $st$-flow with minimum
 energy is unique, and describes the electric current going through that edge
 if the graph represents a network of unit resistors and we put a potential
 difference between $s$ and $t$. The minimum energy flow has several other
 interpretations and connections to other graph properties. For reference,
 and for those who would like to build their intuition for this object, we
 have collected some of these relationships in \cref{app:Flow}.

\paragraph{Graph access} 
We turn graph problems into oracle problems by
 letting a string $x\in\{0,1\}^m$ specify a subgraph $G(x)$ of $G$. In
 particular, we associate each edge $e\in E$ with a number in $[m]$. Then,
 given a string $x\in\{0,1\}^m,$ let $G(x)=(V,E(x))$ be the subgraph of $G$
 that contains an edge $e\in E$ if $e$ is associated with the integer $i\in
 [m]$ and $x_i=1$, where $x_i$ is the $i$th bit of $x$. In this oracle
 problem, one is given access to an oracle $O_x$ for $x$ (or classically,
 given the ability to query the values of the bits of $x$ one at a time), and a
 description of the parent graph $G$ along with the association between
 bits of $x$ and edges of $G$, and the goal is to determine something
 about the graph $G(x)$ using as few queries as possible. Let $E_i\subset E$
 be the set of edges associated with the $i$th bit of $x.$ When not specified otherwise,
 one should assume that $m=|E|$, and then associate each edge of $G$ uniquely with a bit of the input string. In
 this case, when $G$ is the complete graph, $O_x$ is equivalent to query access to the adjacency
 matrix of a graph. When we consider subgraphs of the original graph
 (like $G^-_u$), we assume that the edges are associated with the same
 indices as in the original graph, unless otherwise specified.

Most of the applications in this paper are related to the problem of detecting a path between $s$ and $t$ -- more commonly called
$st$-connectivity. We define $st\textsc{-conn}_G(x)\coloneqq1$ if $s$ and $t$
are connected in $G(x)$, and $0$ otherwise. The following span program, which
we denote by ${\mathcal P}_{G_{st}}$, first
introduced in Ref.~\cite{KW93} and used in the quantum setting in 
Ref.~\cite{belovsSpanProgramsQuantum2012}, decides $st\textsc{-conn}_G(x)$:
for a graph $G=(V,E)$, where $m=|E|$, define the span program ${\mathcal P}_{G_{st}}$ as:
\begin{equation}
\begin{split}
&\forall i\in [m], H_{i,1} = \mathrm{span}\{\ket{(u,v)}:\{u,v\}\in E_i\}, H_{i,0} =\emptyset\\
&\mathcal{V} = \mathrm{span}\{\ket{v}:v\in V(G)\}\\
&\ket{\tau} = \ket{s} - \ket{t}\\
&\forall (u,v)\in\overrightarrow{E}:\; A\ket{u,v} = \ket{u} - \ket{v}.
\end{split}
\label{eq:st-conn-span-program}
\end{equation}

For $\mathcal{P}_{G_{st}}$, the negative approximate 
witness size is bounded by
$\widetilde{W}_-=O(n^2)$ \cite{itoApproximateSpanPrograms2019}. 
If $s$ and $t$ 
are connected in $G(x)$, the optimal positive 
witness of $x$ in ${\cal P}_{G_{st}}$ is 
\cite{belovsSpanProgramsQuantum2012,jarretQuantumAlgorithmsConnectivity2018} 
\begin{equation}\label{eq:optimal_flow_state}
\ket{\theta^*}=\frac{1}{{2}}\sum_{e\in \overrightarrow{E}}\theta^*(e)\ket{e},
\end{equation}
where $\theta^*$ is the $st$-unit flow with minimal energy, so by 
\cref{def:posWitness,def:effRes}, $w_+({\cal P}_{G_{st}},x)=\frac{1}{2}R_{s,t}(G(x)).$

One of our main applications is to apply our witness state generation 
algorithm to the span program ${\cal P}_{G_{st}}$, in which case, we
produce a quantum state close to $\ket{\theta^*}/\|\ket{\theta^*}\|$ where
$\theta^*$ is the optimal unit $st$-flow on $G(x)$. If we were to create
$\ket{\theta^*}/\|\ket{\theta^*}\|$ exactly, and then measure
in the standard basis, the probability that we obtain the edge $e$
is $\theta^*(e)^2/(2R_{s,t}(G(x)))$. Let $\q$ denote the 
distribution such that for $\forall e\in \overrightarrow{E}$, 
\begin{equation}\label{eq:optimal_distribution}
\q(e)=\theta^*(e)^2/(2R_{s,t}(G(x))).
\end{equation}

Additionally, this optimal flow $\theta^*$ is a convex combination of (self-avoiding) $st$-paths, as we prove in \cref{app:Flow}:
\begin{restatable}{lemma}{flowpath}\label{lem:flow-paths}
An \emph{$st$-path} in $G(x)$ is a sequence of \emph{distinct} vertices $\vec{u}=(u_0,\dots,u_{\ell})$ such that $s=u_0$, $t=u_{\ell}$, and for all $i\in [\ell]$, $(u_{i-1},u_i)\in \overrightarrow{E}(G(x))$. From $\vec{u}$, we define
\begin{equation}
\ket{\rho_{\vec{u}}} = \frac{1}{\sqrt{2}}\sum_{i=0}^{\ell-1}(\ket{u_i,u_{i+1}} - \ket{u_{i+1},u_i})
\end{equation}
and refer to all such states as \emph{$st$-path states of $G(x)$}. 
Then if $\ket{\theta^*}$ is the optimal positive witness for $x$ in ${\cal P}_{G_{s,t}}$,
it is a linear combination of $st$-path states in $G(x)$. 
\end{restatable}

A final pair of tools we use are a quantum algorithm that decides 
$st$-$\textsc{conn}_G(x)$ with fewer queries in the case of small effective resistance, without
knowing the effective resistance ahead of time, and a quantum algorithm for estimating the effective resistance:
\begin{lemma}[\cite{anderson2020leverage}]\label{cor:st-conn}
Fix $\delta>0$ and a family of $n$-vertex graphs $G$ with vertices $s$ and $t$.
Then there is a quantum algorithm \emph{\texttt{PathDetection}}$(O_x,G,s,t,\delta)$ such that, 
\begin{enumerate}
    \item The algorithm returns $st$-$\textsc{conn}_G(x)$ with probability $1-O(\delta)$.
    \item On input $x$, the algorithm uses $O\left(n\sqrt{R_{s,t}(G(x))}\log\left(\frac{n}{R_{s,t}(G(x))\delta}\right)\right)$ expected queries if $st$-$\textsc{conn}_G(x)=1$, and
     $O\left(n^{3/2}\log{1/\delta}\right)$ expected queries if $st$-$\textsc{conn}_G(x)=0$.
\end{enumerate}
\end{lemma}

\begin{lemma}[\cite{itoApproximateSpanPrograms2019}]\label{thm:witness-est}
Fix $\delta>0$ and a family of $n$-vertex graphs $G$ with vertices $s$ and $t$.
Then there is a quantum algorithm \emph{\texttt{WitnessSizeEst}}$(O_x,G,s,t,\epsilon,\delta)$ that, on input $x$ such that $st$-$\textsc{conn}_G(x)=1$,
with probability $1-\delta$, outputs an estimate $\hat{R}$ for 
$R_{{s,t}}(G(x))$ such that
\begin{equation}
\left|\hat{R}-R_{{s,t}}(G(x))\right|\leq \epsilon R_{{s,t}}(G(x)),
\end{equation}
using $\widetilde{O}\left(\sqrt{\frac{R_{{s,t}}(G(x))n^2}{\epsilon^3}}\log(1/\delta)\right)$ expected queries; and on input $x$ such that $st$-$\textsc{conn}_G(x)=0$, uses at most $\widetilde{O}\left((n/\epsilon)^{3/2}\log(1/\delta)\right)$.
\end{lemma}

\cref{thm:witness-est} is a special case of \cite[Theorem 3.8]{itoApproximateSpanPrograms2019}, which gives an algorithm for estimating the quantity $w_+(x)$ from \emph{any} span program. If we apply this construction with the span program ${\cal P}_{G_{s,t}}$, we can estimate its positive witness sizes, which are precisely $\frac{1}{2}R_{s,t}(G(x))$. The algorithm described in \cite[Theorem 3.8]{itoApproximateSpanPrograms2019} assumes that the input is a 1-input to $st$-$\textsc{conn}_G(x)$, but can easily be modified to always stop after at most $\widetilde{O}\left((n/\epsilon)^{3/2}\log(1/\delta)\right)$ steps, regardless of the input, since $R_{s,t}(G(x))\leq n$. The algorithm as stated also only works with bounded error, but the success probability can be amplified to $1-\delta$ by repeating $\log(1/\delta)$ times and taking the median estimate.


\section{Witness Generation}\label{sec:witness-gen}

Our main technical result, on generating span program witness states is the following:
\begin{theorem}\label{thm:witness_generation}
Given a span program $\mathcal{P}$ that decides a function $f$, 
and constants $\epsilon,\delta$, there is an algorithm (\cref{alg:state_generation}) that, given as input an oracle $O_x$ such that
$f(x)=1$ with optimal positive witness $\ket{w}$,  outputs a state $\ket{\hat{w}}/\|\ket{\hat{w}}\|$ such that $\norm{\ket{w}/\sqrt{w_+(x)}-\ket{\hat{w}}/\|\ket{\hat{w}}\|}^2\leq O(\epsilon)$ with probability $1-O(\delta)$, and uses $\widetilde{O}\left(\sqrt{\frac{w_+(x)\widetilde{W}_-}{\epsilon}}\log\left(\frac{1}{\delta}\right)\right)$ expected queries to $O_x$.
\end{theorem}

For comparison, a span program algorithm can \emph{decide} $f$ with bounded
error in expected query complexity $\widetilde{O}\left(\sqrt{w_+(x)\widetilde{W}_-}\right)$, so \cref{thm:witness_generation} gives a matching complexity for
generating a witness state. As we will see in \cref{sec:edge-finding}, in the case
of the span program ${\cal P}_{G_{s,t}}$ for $st$-connectivity on subgraphs
of $G$, this implies that we can sample an $st$-path edge in the same
complexity used by the span program algorithm to decide if an $st$-path
exists. 

A key subroutine for our witness state generation algorithm will be quantum phase
estimation. In quantum phase estimation one implements a controlled version of
a unitary $U$ acting on a Hilbert space $\mathcal{H}_A$ on an input state
$\ket{\psi}\in \mathcal{H}_A$. The state $\ket{\psi}$ can be decomposed into
its eigenbasis with respect to $U$ as $\ket{\psi}=\sum_{i}\alpha_i\ket
{\lambda_i}$, where $U\ket{\lambda_i}=e^{i\phi_i \pi}$ and we say $\phi_i$ is
the phase of the state $\ket{\lambda_i}$. Then when phase estimation is
performed with precision $\Theta$ the probability that you measure a phase of
$0$ after the phase estimation procedure is approximately given by 
$\sum_{i:|\phi_i|\leq \Theta }|\alpha_i|^2$, and the non-normalized state that results after
measuring a phase of $0$ is approximately 
$\sum_{i:|\phi_i|\leq \Theta }\alpha_i\ket{\lambda_i}$. In other words, phase
estimation can be used to project into the low phase space (with phase less
than $\Theta$) with probability that depends on the amount of amplitude the
original state had on low-phase eigenstates. For an accuracy parameter
$\epsilon,$ the number of uses of $U$ in phase estimation scales as 
$O\left(\frac{1}{\Theta}\log\frac{1}{\epsilon}\right)$. A more rigorous description of the guarantees of phase estimation is given below in \cref{lem:phase_det}.

The basic idea of the algorithm that we use to prove \cref{thm:witness_generation} 
is to apply phase estimation with a unitary $\U{\mathcal{P}}{x}{\alpha}$, (which can be implemented with access to an oracle $O_x$ and depends on a
span program $\mathcal{P}$, and a positive real parameter $\alpha$), on a state
$\ket{\hat{0}}$. We show that the eigenspectrum of $\ket{\hat{0}}$ relative to $\U{\mathcal{P}}{x}{\alpha}$ decomposes into two states, $\ket{\hat{0}}\oplus\frac{1}{\alpha}\ket{w}$, which is a $0$-phase eigenstate of $\U{\mathcal{P}}{x}{\alpha}$,
and $\tm$, which has small overlap with the low-phase space of $\U{\mathcal{P}}{x}{\alpha}$.

If we do phase estimation with $\U{\mathcal{P}}{x}{\alpha}$ on $\ket{\hat{0}}$ with sufficiently small precision, and then if we measure a phase of $0$, as discussed above,
we will approximately project into the state $\ket{\hat{0}}\oplus\frac{1}{\alpha}\ket{w}$.
From there, if we make the measurement $\{\proj{\hat{0}}, I-\proj{\hat{0}}\}$, and obtain outcome $I-\proj{\hat{0}}$
the state will project into $\ket{w}$, as desired.

Next, there comes a balancing act for our choice of $\alpha$. When $\alpha$ is
too small, $\ket{\hat{0}}$ has small overlap with the span of
$\ket{\hat{0}}\oplus\frac{1}{\alpha}\ket{w}$, so we are not very likely to measure a phase of $0$ when we do phase estimation with $\U{\mathcal{P}}{x}{\alpha}$ on $\ket{\hat{0}}$. However, when
$\alpha$ gets too large, while it becomes very likely to measure a phase of $0$ and thus
obtain the state $\ket{\hat{0}}\oplus\frac{1}{\alpha}\ket{w}$, we will be
unlikely to subsequently measure outcome $I-\proj{\hat{0}}$. 

The sweet spot is when $\alpha\approx \sqrt{w_+(x)}$, in which case both 
measurement outcomes we require have a reasonable probability of occurring. 
Since we don't know $w_+(x)$ ahead of time, we
must first estimate an appropriate value of $\alpha$ to use, which we do by
iteratively testing larger and larger values of $\alpha$.\footnote{There is a similar algorithm in \cite{itoApproximateSpanPrograms2019} that estimates $w_+(x)$, but
it is more precise than we require.} Our test involves
estimating the probability of measuring a phase of $0$ when phase estimation with
$\U{\mathcal{P}}{x}{\alpha}$ is performed on $\ket{\hat{0}}$,
 which we show provides an
estimate of $\alpha/\sqrt{w_+(x)}.$

\subsection{Proof of \cref{thm:witness_generation}}

Before introducing the algorithm we use to prove \cref{thm:witness_generation}, 
we introduce some key concepts, lemmas,
and theorems that will be used in the analysis.

Let
$\tilde{H}=H\oplus \textrm{span}\{\ket{\hat{0}}\},$ 
and $\tilde{H}(x)=H(x)\oplus \textrm{span}\{\ket{\hat{0}}\},$ where 
$\ket{\hat{0}}$ is orthogonal to $H$. Then we define
$\Aa{\alpha} \in \sop L(\tilde{H},{\cal V})$ as
\begin{equation}\label{eq:alphaIntro}
\Aa{\alpha}=\frac{1}{\alpha}\ketbra{\tau}{\hat{0}}-A.
\end{equation}
Let $\Lal{\alpha} \in \sop L(\tilde{H},\tilde{H})$ be the orthogonal
projection onto the kernel of $\Aa{\alpha}$, and let $\PtHx\in \sop
L(\tilde{H},\tilde{H})$ be the orthogonal projector onto $\tilde{H}(x).$
Finally, let $\U{\mathcal{P}}{x}{\alpha}=(2\PtHx-I)(2\Lal{\alpha}-I)$. Note that
$2\PtHx-I$ can be implemented with two applications of $O_x$ \cite[Lemma
3.1]{itoApproximateSpanPrograms2019}, and $2\Lal{\alpha}-I$ can be implemented
without any applications of $O_x$.

We will use parallelized
phase estimation, as described in Ref.
\cite{magniezSearchQuantumWalk2011}, which provides improved error bounds over
 standard phase estimation. In particular, given a unitary $U$ acting on a Hilbert Space
 $\sop H$, a precision $\Theta>0$, and an accuracy $\epsilon>0$, we can
 create a circuit $D(U)$ that implements $O(\log\frac{1}{\epsilon})$ parallel
 copies of the phase estimation circuit on $U$, each to precision $O
 (\Theta)$, that each estimate the phase of a single copy of a state $\ket{\psi}$. That is, $D(U)$ acts on the
 space $\sop H_A\otimes((\mathbb{C}^{2})^{\otimes b})_B$ where $b=O\left
 (\log\frac{1}{\Theta}\log\frac{1}{\epsilon}\right)$, and $A$ 
 labels the input state register, and $B$ labels the
  registers that store the results of the parallel phase estimations.

We use the circuit $D(U)$ to check if an input state has high overlap with the
low-valued eigenphase-space of $U$ \cite{kitaevQuantumMeasurementsAbelian1995,cleveQuantumAlgorithmsRevisited1998,magniezSearchQuantumWalk2011}. 
To characterize the low phase space of a unitary $U$, let
$P_\Theta(U)$ (or just $P_\Theta$ when $U$ is clear from context) be the projection onto $\textrm
{span}\{\ket{u}:U\ket{u}=e^{i\theta}\ket{u}\textrm
{ with }|\theta|\leq\Theta\}$ (the eigenspace of $U$ with eigenphases less than $\Theta$). Then the following lemma provides key properties of parallel phase estimation circuit $D(U)$:
\begin{lemma}[\cite{kitaevQuantumMeasurementsAbelian1995,cleveQuantumAlgorithmsRevisited1998,magniezSearchQuantumWalk2011}]
Let $U$ be a unitary on a Hilbert Space ${\sop H}_A$, and let $\Theta,\epsilon>0$. We call $\Theta$ the precision and $\epsilon$ the accuracy. Then there is a circuit $D(U)$ that acts on the space $\sop H_A\otimes
((\mathbb{C}^{2})^{\otimes b})_B$ for $b=O\left(\log\frac{1}{\Theta}\log\frac{1}{\epsilon}\right)$, and that
uses $O\left(\frac{1}{\Theta}\log\frac{1}{\epsilon}\right)$ controlled calls to
$U$.  Then for any state $\ket{\psi}\in {\sop H}_A$, 
\begin{enumerate}
\item $D(U)(P_0\ket{\psi})_A\ket{0}_B=(P_0\ket{\psi})_A\ket{0}_B$ \label{it:PClower}
\item $\| P_0\ket{\psi}\|^2\leq \|(I_A\otimes \proj{0}_B)D(U)(\ket{\psi}_A\ket{0}_B)\|^2\leq\| P_\Theta\ket{\psi}\|^2+\epsilon$. \label{it:PCupper}
\end{enumerate}
\label{lem:phase_det}
\end{lemma}

 Iterative Quantum Amplitude Estimation is 
 a robust version of amplitude estimation, which
uses repeated applications of amplitude estimation to achieve improved error bounds:
\begin{lemma} [Iterative Quantum Amplitude Estimation \cite{Grinko2021}] Let
 $\delta>0$ and $\mathcal{A}$ be a unitary quantum circuit such that on a state $\ket
 {0}$, $\mathcal{A}\ket{\psi}=\alpha_0\ket{0}\ket{\psi_0}+\alpha_1\ket
 {1}\ket{\psi_1}$. Then there is an algorithm that estimates $|\alpha_0|^2$
 to additive error $\delta$ with success probability at least $1-p$ using 
$O\left(\frac{1}{\delta}\log\left(\frac{1}{p}\log\frac{1}{\delta}\right)\right)$calls to $\mathcal{A}$ and $\mathcal{A}^\dagger$.
\label{lem:ampEst}
\end{lemma}

A key mathematical tool in analyzing span program algorithms is the Effective Spectral Gap Lemma:
\begin{lemma}[Effective Spectral Gap Lemma, \cite{leeQuantumQueryComplexity2011}]
Let $\Pi$ and $\Lambda$ be projections, and let $U=(2\Pi-I)(2\Lambda-I)$ be the
unitary that is the product of their associated reflections. If
$\Lambda\ket{w}=0$, then $\|P_\Theta(U) \Pi\ket{w}\|\leq \frac{\Theta}{2}\|\ket{w}\|.$
\label{spec_gap_lemm}
\end{lemma}

We will need the following relationship between optimal positive witnesses and optimal
negative approximate witnesses:
\begin{theorem}{\em{\cite[Theorem 2.11]{itoApproximateSpanPrograms2019}}} Given a
 span program $\mathcal{P}=(H,\mathcal{V},\ket{\tau},A)$ on $ R^m$ and $x\in R^m$, if $\ket{w}$ is the
 optimal positive witness for $x$ and $\bra{\widetilde{\omega}}$ is an optimal negative
 approximate witness for $x$, then
\begin{equation}
\ket{w}=w_+(x)\PHx(\bra{\widetilde{\omega}} A)^\dagger.
\end{equation}
\label{thm:inverseWitness}
\end{theorem}

As discussed following \cref{thm:witness_generation}, we decompose the state $\ket{\hat{0}}$ into a linear combination of two orthogonal states. They are
\begin{align}\label{eq:psi_defs}
\tp=&\ket{\hat{0}}+\frac{1}{\alpha}\ket{w},\nonumber\\
\tm=&\ket{\hat{0}}-\frac{\alpha}{w_+(x)}\ket{w},
\end{align}
so we can write $\ket{\hat{0}}$ as 
\begin{equation}\label{eq:a_defs}
\ket{\hat{0}}=a_0\tp+a_+\tm,
\quad
\mbox{where}
\quad
a_0=\frac{1}{1+\frac{w_+(x)}{\alpha^2}}, \qquad a_+=\frac{1}{1+\frac{\alpha^2}{w_+(x)}}.
\end{equation}

We first show that $\tp$ is a 0-phase eigenvector of 
$\U{\mathcal{P}}{x}{\alpha}$. Note that 
$\Aa{\alpha}\tp=\frac{1}{\alpha}(\ket{\tau}-\ket{\tau})=0$
(see \cref{eq:alphaIntro}), so recalling that $\Lal{\alpha}$ is the orthogonal
projector onto the kernel of $\Aa{\alpha}$, we have
$\Lal{\alpha}\tp=\tp$. Furthermore, since $\Pi_x$ is the orthogonal
projector onto $\tilde{H}(x)=H(x)\oplus\mathrm{span}\{\ket{\hat0}\}$, it follows that $\Pi_x\tp=\tp$, where we use that $\ket{w}$ is a
positive witness, so $\ket{w}\in H(x)$. Thus $\U{\mathcal{P}}{x}{\alpha}\tp=\tp.$

On the other hand $\tm$ has low overlap with $P_\Theta(\U{\mathcal{P}}{x}{\alpha})$ for small enough $\Theta$ and $\alpha$, as the following lemma shows.

\begin{lemma}
If $\alpha^2\geq 1/\widetilde{W}_-$, then $\|P_\Theta(\U{\mathcal{P}}{x}{\alpha})\tm\|\leq \Theta \alpha\sqrt{\widetilde{W}_-}$.
\label{lem:tm_low_overlap}
\end{lemma}
\begin{proof}
Let $\bra{\widetilde{\omega}}$ be an optimal negative approximate witness for $x$ (see \cref{def:app_negWit}), and let
\begin{equation}\label{eq:vdef}
\ket{v}=\ket{\hat{0}}-\alpha(\bra{\widetilde{\omega}} A)^\dagger.
\end{equation}
Using \cref{thm:inverseWitness} and the fact that $\Pi_x\ket{\hat{0}}=\ket{\hat{0}}$, we have that 
\begin{equation}\label{eq:vprojection}
\Pi_x\ket{v} = \ket{\hat 0} - \alpha \Pi_{H(x)}(\bra{\widetilde\omega} A )^\dagger = \ket{\hat 0} - \alpha\frac{\ket{w}}{w_+(x)} 
=\tm.
\end{equation}
Now we will show $\Lal{\alpha}\ket{v}=0$. Let $\ket{k}$ be in the kernel of $\Aa{\alpha}$, so $\Aa{\alpha}\ket{k}=0$. Using \cref{eq:alphaIntro} and rearranging,
\begin{equation}
A\ket{k}=\frac{1}{\alpha}\ket{\tau}\braket{\hat{0}}{k}.
\label{eq:kernalProp}
\end{equation} 
Then
\begin{equation}
\begin{split}
\braket{v}{k}&=\braket{\hat{0}}{k}-\alpha\bra{\widetilde{\omega}} A\ket{k}\\
&=\braket{\hat{0}}{k}-\braket{\hat{0}}{k}\braket{\widetilde{\omega}}{\tau}\\
&=0
\end{split}
\end{equation}
where we have used \cref{eq:vdef,eq:kernalProp} and the properties of optimal negative approximate witnesses. 
Thus $\ket{v}$ is orthogonal to any element of the kernel of $\Aa{\alpha}$, so 
$\Lal{\alpha}\ket{v}=0$.

Now we can apply \cref{spec_gap_lemm} to $\ket{v}$ to get:
\begin{equation}
\begin{split}
\|P_\Theta(\U{\mathcal{P}}{x}{\alpha})\tm\|&=\|P_\Theta(\U{\mathcal{P}}{x}{\alpha})\Pi_x\ket{v}\|\\
 &\leq \frac{\Theta}{2}\norm{\ket{v}}\\
&=\frac{\Theta}{2}\sqrt{1+\alpha^2\widetilde{w}_-(x,\mathcal{P})}\\
&\leq \Theta \alpha\sqrt{\widetilde{W}_-},
\end{split}
\end{equation}
where in the first line we have used \cref{eq:vprojection}, and in the last,
our assumption that $\alpha^2\widetilde{W}_-\geq 1$.
\end{proof}

\begin{corollary}
\label{cor:tm_no_0}
$\|P_0(U(\mathcal{P},x,\alpha))\tm\|=0$.
\end{corollary}

\begin{proof}
Apply \cref{lem:tm_low_overlap} with $\Theta$ set to $0$.
\end{proof}

\noindent To prove \cref{thm:witness_generation}, we analyze the following algorithm:

\begin{algorithm}[H]
    \DontPrintSemicolon
    \SetKwInOut{Input}{Input}
    \SetKwInOut{Output}{Output}
    \SetKwRepeat{Do}{do}{while}
    \Input{Error tolerance $\delta$, accuracy $\epsilon$, span program $\mathcal{P}$ that decides a function $f$, oracle $O_x$}
    \Output{A quantum state $\ket{\hat{w}}/\|\ket{\hat{w}}\|$ such that for the optimal positive witness $\ket{w}$ for $x$, $\|\ket{w}/\sqrt{w_+(x)}-\ket{\hat{w}}/\|\ket{\hat{w}}\|\|^2\leq O(\epsilon)$ with probability $1-O(\delta)$}
    $\epsilon'\gets \min\{\epsilon,1/96\}; \quad T\gets\left\lceil\log\sqrt{W_+\widetilde{W}_-}\right\rceil$;\quad $p\gets \min\left\{\delta/\log(W_+\widetilde{W}_-),1/\sqrt{W_+\widetilde{W}_-}\right\}$ \label{line:defs}\;
    \tcp{Probing Stage}
    \For{$i= 0$ \KwTo $T$ }{
     $\alpha\gets2^i/\sqrt{\widetilde{W}_-}$\;
    $\hat{a}\gets$ Iterative Amplitude Estimation (\cref{lem:ampEst}) estimate (with probability of failure $p$ and additive error $1/48$) of the probability of outcome $\ket{0}_B$ in register $B$ when  $D(\U{\mathcal{P}}{x}{\alpha})$ (see \cref{lem:phase_det}) acts on $\ket{\hat{0}}_A\ket{0}_B$ with error $\epsilon'$, precision $\sqrt{\frac{\epsilon'}{\alpha^2\widetilde{W}_-}}$\label{line:iter_est}\;
      \lIf{$\frac{15}{48}\leq\hat{a}\leq\frac{35}{48}$}
      {
        Break \label{line:break}
      }
    }
    \tcp{State Generation Stage}
    \For{$j=1$ \KwTo $\log(1/\delta)$}
    {
    Apply $D(\U{\mathcal{P}}{x}{\alpha})$ to $\ket{\hat{0}}_A\ket{0}_B$ with error $\epsilon'$, precision $\sqrt{\frac{\epsilon'}{\alpha^2\widetilde{W}_-}}$\;
    Make a measurement with outcome $M=\{(I-\proj{\hat{0}})_A\otimes \proj{0}_B\}$ on the resultant state \;
    \If{Measure outcome $M$}{
      Return the resultant state\;
    }
    }
    Return ``failure''\;

    \caption{\texttt{WitnessGeneration}$(\mathcal{P},O_x,\delta,\epsilon)$}
    \label{alg:state_generation}
\end{algorithm}

To analyze \cref{alg:state_generation}, will need the following lemma and corollary. In \cref{alg:state_generation}, we estimate the probability of measuring
the outcome $\ket{0}$ in the $B$ register after doing phase estimation. In the following lemma, 
we prove this probability is closely related to $a_0$ from \cref{eq:a_defs}.
\begin{lemma}\label{lemm:aplusProb} 
Applying $D(\U{\mathcal{P}}{x}{\alpha}))$ with error $\epsilon$ and precision 
 $\sqrt{\frac{\epsilon}{\alpha^2\widetilde{W}_-}}$ 
 (see \cref{lem:phase_det}) to input state 
 $\ket{\hat{0}}_A\ket{0}_B$ for 
 $\alpha\geq 1/\sqrt{\widetilde{W}_-}$ results in  the outcome 
 $\ket{0}$ in the $B$ register with
 probability in the range $[a_0,a_0+2\epsilon].$
\end{lemma}

\begin{proof}
Throughout the proof, let $U=U(\mathcal{P},x,\alpha)$. 
The probability that we measure $\ket{0}$ in register $B$ after we apply  
$D(U)$ with error $\epsilon$ 
and precision $\Theta$ to $\ket{\hat{0}}_A\ket{0}_B$ is, by \cref{lem:phase_det} \cref{it:PCupper}, at most
\begin{equation}\label{eq:bound1}
\|P_\Theta(U)\ket{\hat{0}}\|^2+\epsilon= \|a_0P_\Theta(U)\tp+a_+P_\Theta(U)\tm\|^2+\epsilon,
\end{equation}
by \cref{eq:a_defs}. Now $P_\Theta(U)\tp$ and $P_\Theta(U)\tm$ are orthogonal, since
\begin{equation}
\tpb P_\Theta(U) P_\Theta(U) \tm=\tpmbk=0,
\end{equation}
where we've used that $P_\Theta(U)\tp=\tp$ and that $\tp$ and $\tm$ are orthogonal. 
Continuing from \cref{eq:bound1} and using the orthogonality condition, we have, 
using $\Theta=\sqrt{\frac{\epsilon}{\alpha^2\widetilde{W}_-}}$,
\begin{align}
\|P_\Theta(U)\ket{\hat{0}}\|^2+\epsilon&=a_0^2\|P_\Theta(U)\tp\|^2+a_+^2\|P_\Theta(U)\tm\|^2+\epsilon\nonumber\\
&\leq a_0^2\|\tp\|^2+a_+^2\Theta^2\alpha^2\widetilde{W}_-+\epsilon & \mbox{by \cref{lem:tm_low_overlap}, since $\alpha^2\widetilde{W}_-\geq 1$}\nonumber\\
&\leq a_0+a_+^2\epsilon+\epsilon\nonumber\\
&\leq a_0+2\epsilon,
\label{eq:up_bound}
\end{align}
where we have used that $\|\tp\|^2=1/{a_0}$, and $a_+\leq 1$ (see \cref{eq:a_defs}) .

By \cref{lem:phase_det} \cref{it:PCupper}, the probability that we measure
$\ket{0}$ in register $B$ after applying $D(\U{\mathcal{P}}{x}{\alpha})$ on 
$\ket{\hat{0}}_A\ket{0}_B$ with error ${\epsilon}$ and any precision is at least
\begin{align}
\|P_0(U)\ket{\hat{0}}\|^2=\|a_0P_0(U)\tp+a_+P_0(U)\tm\|^2=a_0^2\|\tp\|^2=a_0,
\label{eq:low_bound}
\end{align}
where we have used \cref{cor:tm_no_0}.
\end{proof}

\begin{corollary} In \cref{alg:state_generation}, if in an iteration of the
 Probing Stage, Iterative Amplitude Estimation does not fail at \cref{line:iter_est} and
 subsequently causes a break at \cref{line:break}, then
\begin{equation}\label{eq:confined}
a_0\in\left[\frac{1}{4},\frac{3}{4}\right],
\qquad \frac{a_0^2w_+(x)}{\alpha^2}\in\left[\frac{3}{16},\frac{1}{4}\right].
\end{equation}
\label{cor:confined}
\end{corollary}

\begin{proof}
If Iterative Amplitude Estimation does not fail at \cref{line:iter_est} and causes a break
at \cref{line:break}, then we have an estimate $\hat{a}$ that is in the range 
$[\frac{15}{48},\frac{35}{48}]$. Thus, because of the additive error of $1/48$ in
Iterative Amplitude Estimation, the probability of measuring outcome $\ket
{0}_B$ is in the range $[\frac{14}{48},\frac{36}{48}]$. By 
\cref{lemm:aplusProb}, this same probability is in the range $[a_0,a_0+2\epsilon']$, 
so in particular these two ranges overlap. Thus, since we choose $2\epsilon'$ to be at most
$1/48$, we have that 
\begin{equation}\label{eq:a+confined}
a_0\in\left[\frac{13}{48},\frac{36}{48}\right]\subset \left[\frac{1}{4},\frac{3}{4}\right].
\end{equation}
Using $a_0=(1+\frac{w_+(x)}{\alpha^2})^{-1}$
 (see \cref{eq:a_defs}), this implies the stated ranges for 
 $\frac{a_0^2w_+(x)}{\alpha^2}=a_0(1-a_0)$.
\end{proof}

Now we prove the main performance guarantees of \cref
{alg:state_generation}, bounding the success probability and the expected
query complexity, thus proving \cref{thm:witness_generation}.

\begin{proof} [Proof of \cref{thm:witness_generation}]
Letting $U=U(\mathcal{P},x,\alpha)$, we analyze \cref{alg:state_generation}. 
We first show that the algorithm will produce the desired state if both the
Probing Stage and the State Generation stage are successful. Then we will
analyze the probability of this occurring, in order to bound the success
probability of the algorithm.

We say the Probing Stage is successful if in some iteration, Iterative Amplitude
estimation, having not failed thus far, does not fail and then triggers a
break at Line 6, in which case we can apply \cref{cor:confined}. Under these
assumptions, we consider the outcome of a successful State Generation stage,
when we achieve the measurement outcome 
$M=(I-\proj{\hat{0}})_A\otimes \proj{0}_B$. The non-normalized state $\ket{\hat{w}}$ that is produced upon
measurement outcome $M$  is
\begin{align}\label{eq:wAnanlysis}
\ket{\hat{w}}&=(I-\proj{\hat{0}})_A\otimes \proj{0}_B D(U)\ket{\hat{0}}_A\ket{0}_B\nonumber\\
&=a_0(I-\proj{\hat{0}})_AD(U)\tp_A\ket{0}_B+a_+(I-\proj{\hat{0}})_A\otimes \proj{0}_B D(U)\tm_A\ket{0}_B\nonumber\\
&=a_0(I-\proj{\hat{0}})_A\tp_A\ket{0}_B+a_+(I-\proj{\hat{0}})_A\otimes \proj{0}_B D(U)\tm_A\ket{0}_B\nonumber\\
&=\frac{a_0}{\alpha}\ket{w}_A\ket{0}_B+\underbrace{a_+{(I-\proj{\hat{0}})_A\otimes \proj{0}_B D(U)\tm_A\ket{0}_B}}_{\eqqcolon \ket{\xi}},
\end{align}
where in the final equality, 
we used \cref{lem:phase_det} \cref{it:PClower}, since $P_0(U)\tp=\tp$.

We would like to bound $\Delta$, where
\begin{align}
\Delta\coloneqq \norm{\frac{\ket{\hat w}}{\norm{\ket{\hat w}}}-\frac{\ket{w}_A\ket{0}_B}{\sqrt{w_+(x)}}} 
&= \norm{\frac{\frac{a_0}{\alpha}\ket{w}_A\ket{0}_B+\ket{\xi}}{\norm{\ket{\hat w}}}-\frac{\ket{w}_A\ket{0}_B}{\sqrt{w_+(x)}}}\nonumber\\
&\leq \abs{\frac{a_0}{\alpha \norm{\ket{\hat w}}}-\frac{1}{\sqrt{w_+(x)}}}\norm{\ket{w}}+\frac{\norm{\ket{\xi}}}{\norm{\ket{\hat w}}}&\mbox{by triangle ineq.}\nonumber\\
&\leq \abs{\frac{a_0\sqrt{w_+(x)}}{\alpha \norm{\ket{\hat w}}}-1}+\frac{\norm{\ket{\xi}}}{\norm{\ket{\hat w}}}. \label{eq:Delta1}
\end{align}

To bound $\|\ket{\xi}\|$, we have
\begin{equation}\label{eq:etaAnalysis}
\begin{split}
\norm{\ket{\xi}}^2=a_+^2\left\|(I-\proj{\hat{0}})_A\otimes \proj{0}_B D(U)\tm_A\ket{0}_B\right\|^2
\leq&\|I_A\otimes \proj{0}_B D(U)\tm_A\ket{0}_B\|^2\\
\leq&\|P_\Theta\tm\|^2+\epsilon'\\
\leq& \Theta^2 \alpha^2{\widetilde{W}_-}+\epsilon'
\leq 2\epsilon',
\end{split}
\end{equation} 
where the first inequality is because a projection can only
 decrease the norm of a vector, and $a_+\leq 1$; the second inequality is
 from by \cref{lem:phase_det} \cref{it:PCupper}, and the third inequality
 comes from \cref{lem:tm_low_overlap} and our choice of $\Theta$. 

Next, to bound $\|\ket{\hat{w}}\|$, 
we use the triangle inequality on the final line of \cref{eq:wAnanlysis}, and \cref{eq:etaAnalysis} to get
\begin{equation}\label{eq:norm-hat-w}
\frac{a_0\sqrt{w_+(x)}}{\alpha}-\sqrt{2\epsilon'}\leq \norm{\ket{\hat w}} \leq \frac{a_0\sqrt{w_+(x)}}{\alpha}+\sqrt{2\epsilon'}.
\end{equation}
By our choice of $\epsilon'$, we have $2\epsilon'\leq 1/48$, and also 
applying \cref{cor:confined} to \cref{eq:norm-hat-w}, we have
\begin{equation}\label{eq:norm-hat-w-numbers}
\frac{1}{4}< \sqrt{3/16}-\sqrt{1/48}\leq \norm{\ket{\hat w}} \leq\sqrt{1/4}+\sqrt{1/48}< \frac{3}{4}.
\end{equation}
Rearranging \cref{eq:norm-hat-w} and applying \cref{eq:norm-hat-w-numbers}, we have
\begin{align}\label{eq:abs-bound}
\left|\frac{a_0\sqrt{w_+(x)}}{\alpha\norm{\ket{\hat w}}}-1\right|\leq 
\frac{\sqrt{2\epsilon'}}{\norm{\ket{\hat w}}}.
\end{align}
Then plugging \cref{eq:etaAnalysis,eq:norm-hat-w-numbers,eq:abs-bound} 
into \cref{eq:Delta1} we have:
\begin{align}
\Delta\leq \frac{2\sqrt{2\epsilon'}}{\norm{\ket{\hat w}}} < 8\sqrt{2\epsilon'} = O(\epsilon). 
\end{align}

Now we analyze the probability that both the Probing Stage and State Generation Stage are successful, resulting in the state 
$\ket{\hat w}/\norm{\ket{\hat w}}$ as in \cref{eq:Delta1}. First note that
there is a value of $\alpha$ (if we iterate in the Probing Stage long
enough), that will cause us to break out of the Probing Stage if Iterative
Amplitude Estimation does not fail. In particular, when $w_+
(x)/\alpha^2\in\left[1/2,2\right]$, then
from \cref{eq:a_defs} $a_0\in \left[1/3,2/3\right]$. 
 Thus by \cref
{lemm:aplusProb} and since $2\epsilon'\leq 1/48$, the probability of of
outcome $\ket{0}_B$ is in $[16/48,33/48]$, which in Line 5 causes us to leave
the Probing Stage if Iterative Amplitude Estimation does not fail.
This occurs for some value of $\alpha$,
as we are doubling $\alpha$ at each iteration of the Probing Stage, causing $w_+(x)/\alpha^2$ to decrease, and initially we 
have $w_+(x)/\alpha^2 = w_+(x)\widetilde{W}_-\geq 1$.\footnote{To see that $w_+(x)\widetilde{W}_-\geq 1$, 
let $N_+=\min\{\norm{\ket{w}}^2:A\ket{w}=\ket{\tau}\}$, and 
$N_-=\min\{\norm{\bra{\omega}A}^2:\braket{\omega}{\tau}=1\}$. Then $w_+(x)\geq N_+$, 
and $\widetilde{W}_-\geq N_-$, and by \cite[Section 2.4]{itoApproximateSpanPrograms2019}, $N_+N_-=1$.}

Thus if no error occurs, the condition of Line 5 will be satisfied after some
number $L$ of rounds such that $L\in O(\log({w_+(x)\widetilde{W}_-}))=O(\log (W_+\widetilde{W}_-))$. As the
probability of failing a single Iterative Amplitude Estimation round is
$p\leq \delta/\log(W_+\widetilde{W}_-)$ (see \cref{line:defs}), the
probability of leaving the Probing Stage when Line 5 is satisfied (rather than
before or after) is at least
\begin{equation}\label{eq:successPS}
(1-p)^{L}=1-O(\delta).
\end{equation}

Assuming that we have successfully left the Probing Stage without failure, we
next calculate the probability of getting a measurement outcome $M$ during
the at most $\log(1/\delta)$ iterations of the State Generation Stage. The
probability of getting outcome $M$ is lower bounded by (from \cref{eq:norm-hat-w-numbers})
\begin{equation}
\|\ket{\hat{w}}\|^2\geq 1/16.
\end{equation}
Thus the probability of success in the State Generation Stage is
\begin{equation}\label{eq:successSG}
1-(15/16)^{\log(1/\delta)}=1-O(\delta).
\end{equation}

Combining \cref{eq:successPS,eq:successSG}, our probability of successfully producing a state $\ket{\hat w}/\norm{\ket{\hat w}}$ as in \cref{eq:Delta1} is
\begin{equation}
(1-O(\delta))(1-O(\delta))=1-O(\delta).
\end{equation}

To calculate the expected query complexity, we first note that if we terminate in round
$t\in \{0,\dots,\ceil{\log \sqrt{W_+}}\}$ of the Probing Stage, we use
\begin{equation}\label{eq:tcost}
\begin{split}
&\sum_{i=0}^tO\left(\frac{2^{i}}{\sqrt{\epsilon}}\log\left(\frac{1}{\epsilon}\right)\log\left(\frac{1}{p}\right)\right)
+O\left(\log\left(\frac{1}{\delta}\right)\frac{2^{t}}{\sqrt{\epsilon}}\log\left(\frac{1}{\epsilon}\right)\right)\\
=&O\left(\frac{2^t}{\sqrt{\epsilon}}\log\left(\frac{1}{\epsilon}\right)\log\left(\frac{1}{p\delta}\right)\right)
\end{split}
\end{equation}
queries, which comes from the cost of Iterative Amplitude Estimation
 (\cref{lem:ampEst}) applied to phase estimation (\cref{lem:phase_det}) in
 each round of the Probing Stage up to the $t^{\textrm{th}}$ round, plus the cost
 of phase estimation in the State Conversion Stage.

The probability that we terminate in any round $t$ when we have an estimate
$\hat{a}$ that is not in the range $[\frac{15}{48},\frac{35}{48}]$ is at most
$p$. Using \cref{eq:tcost} the the total contribution to the average query complexity
from all such rounds is at most
\begin{equation}\label{eq:BadRounds}
\begin{split}
\sum_{t=0}^{\lceil\log\sqrt{W_+\widetilde{W}_-}\rceil} O\left(p\frac{2^t}{\sqrt{\epsilon}}\log\left(\frac{1}{\epsilon}\right)\log\left(\frac{1}{p\delta}\right)\right)
= O\left(p\sqrt{\frac{W_+\widetilde{W}_-}{\epsilon}}\log\left(\frac{1}{\epsilon}\right)\log\left(\frac{1}{p\delta}\right)\right).
\end{split}
\end{equation}
where in the sum we have actually included all rounds, not just
 those that satisfy when $\hat{a}$ is not in the range $[\frac{15}{48},\frac
 {35}{48}]$, which is acceptable since we are deriving an upper bound on the expected query complexity.

If we terminate at a round $t^*$ when $\hat{a}$ is in the range 
$[\frac{15}{48},\frac{35}{48}]$, which happens when Iterative Amplitude Estimation 
does not fail at Line 4 and then causes a break at Line 5, from \cref{eq:confined} we have 
$\frac{w_+(x)}{\alpha^2}\in\left[\frac{1}{3},4\right]$, and
$2^{t^*}=\alpha\sqrt{\widetilde{W}_-}$ so 
$\sqrt{w_+(x)\widetilde{W}_-}/2\leq 2^{t^*}\leq \sqrt{3w_+(x)\widetilde{W}_-}$. Because 
we double $\alpha$ at each iteration, there are only a
 constant number of rounds where we will find $\hat{a}$ in the appropriate
 range, and we trivially upper bound the probability of terminating at any such
 round by $1$. Using \cref{eq:tcost}, these rounds add
\begin{equation}\label{eq:GoodRounds}
O\left(\sqrt{\frac{w_+(x)\widetilde{W}_-}{\epsilon}}\log\left(\frac{1}{\epsilon}\right)\log\left(\frac{1}{p\delta}\right)\right)
\end{equation}
to the total expected query complexity.

Combining \cref{eq:BadRounds,eq:GoodRounds}, and using that we set $p$ to be 
$O\left(1/\sqrt{W_+\widetilde{W}_-})\right)$ (\cref{line:defs}), we find the 
expected query complexity is
\begin{equation}\label{eq:av_quer_complexity}
O\left(\sqrt{\frac{w_+(x)\widetilde{W}_-}{\epsilon}}\log\left(\frac{1}{\epsilon}\right)\log\left(\frac{1}{p\delta}\right)\right)=\widetilde{O}\left(\sqrt{\frac{w_+(x)\widetilde{W}_-}{\epsilon}}\log\left(\frac{1}{\delta}\right)\right).\qedhere
\end{equation}
\end{proof}

\section{Graph Applications}\label{sec:applications}

\subsection{Finding an Edge on a Path}\label{sec:edge-finding}
In this section, we consider the problem of finding an edge on an $st$-path in 
$G(x)$, which we denote
$st\mbox{-}\textsc{edge}_G(x)$. That is, given query access to a string $x$ that determines a subgraph $G(x)=(V,E(x))$ of an $n$-vertex graph $G$, as described in \cref{sec:graph-theory} (if $G$ is a complete graph, $x$ is just the adjacency matrix of $G(x)$),  with $s,t\in V$ such that there is at least one path from $s$ to $t$ in $G(x)$, output an edge
$e\in E(x)$ that is on a (self-avoiding) path from $s$ to $t$.

Classically, it is hard to imagine that this problem is much easier than
finding a path, and indeed, in our classical lower bound in \cref
{thm:class_edge_finding} we force the algorithm to learn a complete
path before it can find any edge on the path. However, we find that quantumly,
when there are short or multiple paths, this problem is easier than any path finding algorithms known.
This opens up the possibility of improved quantum algorithms for cases where it is not necessary to know the complete path,
like the $st$-cut set algorithm of \cref{sec:sabotage}.

\begin{restatable}{theorem}{pathEdgeSample}
Fix $p>0$, and a family of $n$-vertex graphs $G$ with vertices $s$ and $t$. There is a quantum algorithm (\cref{alg:edge_finder}) that solves $st\mbox{-}\textsc{edge}_G(x)$ with probability $1-O(p)$ and 
uses $\widetilde{O}\left(\frac{n\sqrt{R_{s,t}(G(x))}}{p}\right)$ expected queries on input $x$. 
More precisely, with probability $1-O(p)$, the algorithm samples from a distribution $\hat{q}$ such the total variation distance between $\hat{q}$ and $\q$ is $O(\sqrt{p})$, where $\q(u,v)$ (defined in \cref{eq:optimal_distribution}) is proportional to $\theta^*(u,v)^2$, where $\theta^*$ is the optimal unit $st$-flow on $G(x)$.
 \label{thm:edge_finder}
\end{restatable}

To obtain this result, we run our
witness state generation algorithm (\cref{alg:state_generation}) using the 
span program for $st$-connectivity, ${\cal P}_{G_{st}}$ and an oracle $O_x$
that defines a graph
$G(x)$
with a path between $s$ and $t$. When successful,
the output will be a quantum state that is approximately proportional
to the optimal flow state, \cref{eq:optimal_flow_state}, which itself is a superposition of edges 
on paths by \cref{lem:flow-paths}. Then from \cref{eq:optimal_distribution}, when we then measure in the standard basis, the probability of obtaining  an edge $e$ should be close to $\q(e)$,
and with high probability, we will measure some edge on a path.

\begin{proof}[Proof of \cref{thm:edge_finder}:]

We analyze \cref{alg:edge_finder}.

\begin{algorithm}[h]
    \DontPrintSemicolon
    \SetKwInOut{Input}{Input}
    \SetKwInOut{Output}{Output}
    \SetKwRepeat{Do}{do}{while}
    \Input{Failure tolerance $p>0$, oracle $O_x$ for the graph $G(x)=(V,E(x))$, $s,t\in V$ such that there is a path from $s$ to $t$. }
    \Output{An output $e$, or ``Failure'', such that with probability $1-O(p)$, $e$ is an edge on a path from $s$ to $t$.}
    $\epsilon\gets p^2$;\quad $\delta\gets p$\;
    $\ket{\hat{\theta}}\gets\texttt{WitnessGeneration}({\cal P}_{G_{st}},O_x,\epsilon,\delta)$ (\cref{alg:state_generation})\;
    \If{$\ket{\hat{\theta}}\neq$ ``Failure''}
    {
     $ e\gets $result of Measuring $\ket{\hat{\theta}}$ in the standard basis\;
    }
    Return ``Failure''

    \caption{\texttt{EdgeFinder}$(O_x,p,G,s,t)$}
    \label{alg:edge_finder}
\end{algorithm}

If \texttt{WitnessGeneration}$({\cal P}_{G_{st}},O_x,\epsilon,\delta)$ (see \cref{alg:state_generation}) does not fail, which
happens with probability $1-O(\delta)=1-O(p)$, then by \cref
{thm:witness_generation},
\begin{equation}\label{eq:created_state}
\ket{\hat{\theta}}=\ket{\theta^*}/\|\ket{\theta^*}\|+\ket{\eta}
\end{equation}
for some $\ket{\eta}$ such that $\|\ket{\eta}\|^2=O(\epsilon)$ and from \cref{eq:optimal_flow_state}, 
$\ket{\theta^*}=\frac{1}{{2}}\sum_{e\in \overrightarrow{E}}\theta^*(e)\ket{e}$
where $\theta^*$ is the optimal unit $st$-flow in $G(x)$, so 
$\|\ket{\theta^*}\|=\sqrt{\mathcal{R}_{s,t}(G(x))}$. 

Let $P_{E(x),s,t}$ be the projection onto the set
of edges in $\overrightarrow{E}(x)$ that are on (self-avoiding) paths from $s$ to $t$. The probability that we measure such 
an edge when we measure $\ket{\hat{\theta}}$ in the standard basis is the square of
\begin{align}
\norm{P_{E(x),s,t}\ket{\hat{\theta}}}
&\geq \norm{P_{E(x),s,t}\ket{\theta^*}/\norm{\ket{\theta^*}}}-\norm{P_{E(x),s,t}\ket{\eta}}=1-O(\sqrt{\epsilon}),
\end{align}
where we have used the triangle inequality, and the fact that $P_{E(x),s,t}\ket{\theta^*}=\ket{\theta^*}$, by \cref{lem:flow-paths}.
Continuing, we have probability
\begin{align}
\norm{P_{E(x),s,t}\ket{\hat{\theta}}}^2 &\geq \left(1-O(\sqrt{\epsilon})\right)^2 = 1 - O(\sqrt{\epsilon}).
\end{align}

Thus our total probability of success of measuring an edge on a path is $(1-O(\delta))(1-O(\sqrt{\epsilon})$. Since we are setting
$\epsilon$ to $p^2$ and $\delta$ to $p$, our total probability of success is $1-O(p)$. 

Let $\hat{q}$ be the output distribution of \cref{alg:edge_finder}. By the relationship between
total variation distance and trace norm, we have that $d(\hat{q},\q)$, the total variation
distance between $\hat{q}$ and $\q$, is at most the trace norm of $\ket{\hat{\theta}}$ and
$\ket{\theta^*}/\|\ket{\theta^*}\|$ (see e.g. \cite{nielsen2010quantum}) so
\begin{align}
d(\hat{q},\q)&\leq \sqrt{1-\abs{\braket{\hat\theta}{\theta^*}/\norm{\ket{\theta^*}}}^2}\nonumber\\
&=\sqrt{1-\abs{\braket{\hat\theta}{\hat{\theta}}-\braket{\hat\theta}{\eta}}^2}\nonumber\\
&\leq \sqrt{1- \left( 1 - \norm{\ket{\eta}} \right)^2}\nonumber\\
&\leq \sqrt{2\norm{\ket{\eta}}}
= O(\epsilon^{1/4}) = O(\sqrt{p}).
\end{align} 
By \cref{thm:witness_generation}, the expected query complexity of \texttt{WitnessGeneration}, and thus \cref{alg:edge_finder} is
\begin{equation}
\widetilde{O}\left(\sqrt{\frac{w_+(x)\widetilde{W}_-}{\epsilon}}\log\left(\frac{1}{\delta}\right)\right)=\widetilde{O}\left(
 \frac{\sqrt{R_{s,t}(G(x))}n}{p}\right)
\end{equation}
where we have used the fact that, for ${\cal P}_{G_{st}}$, $w_+(x)=R_{s,t}(G(x))$ and $\widetilde{W}_-=O(n^2)$ \cite{belovsSpanProgramsQuantum2012,itoApproximateSpanPrograms2019}; and set $\epsilon$ to $p^2$ and $\delta$ to $p$, as in \cref{alg:edge_finder}.
\end{proof}

We can use \cref{thm:edge_finder} to prove the following separation between
the quantum and classical query complexity of finding an edge on a path:
\begin{restatable}{theorem}{ClassEdgeFind}\label{thm:class_edge_finding} 
Let $G=(V,E)$ with $s,t\in V$ be an $n$-vertex complete graph, and suppose we are 
promised that $G(x)$ has a path of length $L$ for $L\in[3,n/4]$ between $s$
 and $t$ ($L$ may depend on $x$ and need not be known ahead of time). Then $st\mbox{-}\textsc{edge}_G(x)$
 can be solved in  $\widetilde{O}(n\sqrt{L})$ expected quantum queries on input $x$, while any classical algorithm has
 query complexity $\Omega(n^2)$.
\end{restatable}

\begin{proof} For the quantum algorithm, we apply \cref{thm:edge_finder} with
 bounded probability of error $p=\Omega(1)$, and use the fact that $R_{s,t}(G)=O
 (L).$

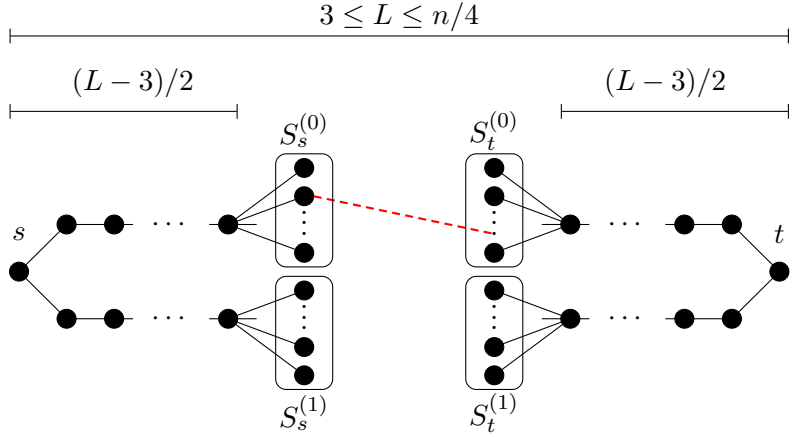
\begin{figure}[ht]
\centering
\begin{tikzpicture}[scale = 1.25]

\filldraw (-4,0) circle (.1);
\filldraw (4,0) circle (.1);

\filldraw (-3.5,.5) circle (.1);
\filldraw (3.5,.5) circle (.1);
\filldraw (-3,.5) circle (.1);
\filldraw (3,.5) circle (.1);
\node at (-2.4,.5) {$\cdots$};
\node at (2.4,.5) {$\cdots$};
\filldraw (-1.8,.5) circle (.1);
\filldraw (1.8,.5) circle (.1);

\filldraw (-3.5,-.5) circle (.1);
\filldraw (3.5,-.5) circle (.1);
\filldraw (-3,-.5) circle (.1);
\filldraw (3,-.5) circle (.1);
\node at (-2.4,-.5) {$\cdots$};
\node at (2.4,-.5) {$\cdots$};
\filldraw (-1.8,-.5) circle (.1);
\filldraw (1.8,-.5) circle (.1);


\filldraw (-1,1.1) circle (.1);
\filldraw (1,1.1) circle (.1);
\filldraw (-1,.8) circle (.1);
\filldraw (1,.8) circle (.1);
\node at (-1,.6) {$\vdots$};
\node at (1,.6) {$\vdots$};
\filldraw (-1,.2) circle (.1);
\filldraw (1,.2) circle (.1);

\filldraw (-1,-1.1) circle (.1);
\filldraw (1,-1.1) circle (.1);
\filldraw (-1,-.8) circle (.1);
\filldraw (1,-.8) circle (.1);
\node at (-1,-.4) {$\vdots$};
\node at (1,-.4) {$\vdots$};
\filldraw (-1,-.2) circle (.1);
\filldraw (1,-.2) circle (.1);

\draw (-3.5,.5)--(-4,0);
\draw (-3.5,-.5)--(-4,0);
\draw (3.5,.5)--(4,0);
\draw (3.5,-.5)--(4,0);

\draw (-3.5,.5)--(-2.8,.5);
\draw (3.5,.5)--(2.8,.5);
\draw (-3.5,-.5)--(-2.8,-.5);
\draw (3.5,-.5)--(2.8,-.5);
\draw (-2,.5)--(-1.5,.5);
\draw (2,.5)--(1.5,.5);
\draw (-2,-.5)--(-1.5,-.5);
\draw (2,-.5)--(1.5,-.5);

\draw (-1.8,.5)--(-1,1.1);
\draw (-1.8,.5)--(-1,.8);
\draw (-1.8,.5)--(-1,.2);

\draw (1.8,.5)--(1,1.1);
\draw (1.8,.5)--(1,.8);
\draw (1.8,.5)--(1,.2);

\draw (-1.8,-.5)--(-1,-1.1);
\draw (-1.8,-.5)--(-1,-.8);
\draw (-1.8,-.5)--(-1,-.2);

\draw (1.8,-.5)--(1,-1.1);
\draw (1.8,-.5)--(1,-.8);
\draw (1.8,-.5)--(1,-.2);

\draw [color=red,densely dashed,thick] (-.9,.8)--(1,.4);

\node at (-4,.4) {$s$};
\node at (4,.4) {$t$};
\node at (-1,1.5) {$S_s^{(0)}$};
\node at (1,1.5) {$S_t^{(0)}$};
\node at (-1,-1.5) {$S_s^{(1)}$};
\node at (1,-1.5) {$S_t^{(1)}$};

\draw[rounded corners] (-.7, 1.25) rectangle (-1.3,.05);  \draw[rounded corners] (.7, 1.25) rectangle (1.3,.05);
\draw[rounded corners] (-.7, -1.25) rectangle (-1.3,-.05);  \draw[rounded corners] (.7, -1.25) rectangle (1.3,-.05);

\draw [|-|](-4.1,2.5) -- (4.1,2.5);
\node at (0,2.7) {$3\leq L\leq n/4$};

\draw [|-|](-4.1,1.7) -- (-1.7,1.7);
\node at (-2.8,2) {$(L-3)/2$};
\draw [|-|](4.1,1.7) -- (1.7,1.7);
\node at (2.8,2) {$(L-3)/2$};

\end{tikzpicture}
\caption{The solid black lines show the edges that are present in $G(x)$ for 
any $x$. In addition, $G(x)$ contains a single edge between a vertex in $S_s^{(b)}$ 
and $S_t^{(b)}$, where $b=\sigma^*_1$, as in the dashed red edge, resulting in a 
single path of length $L$.}\label{fig:findEdge}
\end{figure}

 For the classical lower bound, we reduce the following
problem to path edge finding: Given a string $x$ of $N=2^{\ell}$ bits, 
$(x_{\sigma})_{\sigma\in \{0,1\}^{\ell}}$ such that there is a 
unique $\sigma^*$ with $x_{\sigma^*}=1$, output $\sigma^*_1$. That is,
we would like to
output the first bit of the index of the unique $1$-valued bit of $x$. By an adversary argument similar to a
standard OR lower bound, the bounded error randomized query complexity of
this problem is $\Omega(N)$. We will show how to solve this problem with
an algorithm for finding a path edge 
on a graph like the one depicted in \cref{fig:findEdge}. 

For $x\in \{0,1\}^N$, let $G(x)$ be a graph on $n=\Theta(2^{\ell/2})$ vertices
in which there is a unique $st$-path of length $L$, for some odd $L$, as shown
in \cref{fig:findEdge}. The vertex $s$ is connected by a path of 
length $(L-3)/2$ to a vertex that is additionally connected to a set of 
$2^{(\ell-1)/2}$ vertices, 
$S_s^{(0)}=\{u_{0,\sigma}:\sigma\in\{0,1\}^{(\ell-1)/2}\}$. 
In a symmetric manner, $s$ is also connected by another
disjoint path of length $(L-3)/2$ to a vertex that is additionally connected
to a set of $2^{(\ell-1)/2}$ vertices, 
$S_s^{(1)}=\{u_{1,\sigma}:\sigma\in\{0,1\}^{(\ell-1)/2}\}$. In the same way, $t$ is connected by a pair of
disjoint paths of length $(L-3)/2$ to a pair of vertices, additionally
connected to $S_t^{(0)}=\{v_{0,\sigma}:\sigma\in\{0,1\}^{(\ell-1)/2}\}$ and
$S_t^{(1)}=\{v_{1,\sigma}:\sigma\in\{0,1\}^{(\ell-1)/2}\}$ respectively. All
edges described so far (the black edges in \cref{fig:findEdge}) are always
present in $G(x)$ (we simulate querying the associated input bits by just outputting 1). We now describe edges whose presence in $G(x)$ is
determined by $x$. For $b\in\{0,1\}$, there is a potential edge between every
pair of vertices $u_{b,\sigma}\in {S}_s^{(b)}$ and 
$v_{b,\sigma'}\in S_t^{(b)}$, with the label $x_{b\sigma\sigma'}$, meaning exactly one of these is
present in $G(x)$ -- the one with $\sigma^*=b\sigma\sigma'$. All remaining possible edges are never present in $G(x)$ (we simulate querying their associated input bits by just outputting 0).  

We can find the first bit of $\sigma^*$ by running the edge finding algorithm on $G(x)$. 
Assuming the output is correct, there are the following possibilities:
\begin{enumerate}
\item If the algorithm outputs an edge from the middle part of the graph, then it must be the one labelled by $x_{\sigma^*}$, so $\sigma^*$ is learned entirely.
\item If the algorithm outputs an edge from the left-hand side of the graph, it is on a path between $s$ and $S_s^{(b)}$ for some $b\in\{0,1\}$, and we know that $\sigma^*_1=b$.
\item If the algorithm outputs an edge from the right-hand side of the graph, it is on a path between $t$ and $S_t^{(b)}$ for some $b\in\{0,1\}$, and we know that $\sigma^*_1=b$.
\end{enumerate}
In all cases, we have learned $\sigma^*_1$. This gives a lower bound on path-edge 
finding of $\Omega(N)=\Omega(2^{\ell})=\Omega(n^2)$.
\end{proof}

\subsection{Finding an \texorpdfstring{$st$}{st}-cut set}\label{sec:sabotage}

Given a graph $G(x)$ containing a path from $s$ to $t$, an $st$-cut set is a set
of edges in $G(x)$ such that when those edges are removed from $G(x)$, there
is no longer a path from $s$ to $t$. The $st$-cut set problem is that of finding
an $st$-cut set. This problem has applications to detecting weak points in
networks in order to figure out how to strengthen a network, or conversely,
for sabotaging networks.

We first note that for graphs with a single $st$-path, \cref{thm:edge_finder} can immediately be
used to find an $st$-cut set, since any edge on the path is an $st$-cut set.
However, we can also analyze more complex situations, as the following, in which
we have an upper bound on the effective resistance of the graph, and a lower
bound on the optimal unit $st$-flow  going through any edge in the $st$-cut set:
\begin{theorem}\label{thm:small_cut}
For functions $R,g: \mathbb{N}\rightarrow \mathbb{R}_{> 0}$, let
$G=(V,E)$ with $s,t\in V$ be a family of $n$-vertex graphs, and suppose 
we are additionally promised that $R_{s,t}(G
(x))\leq R(n)$, and there exists an $st$-cut set $C\subseteq E(x)$ such that for each $\{u,v\}\in C$,
 $\theta^*(u,v)^2\geq g(n)$ where
$\theta^*$ is the optimal unit $st$-flow in $G(x)$. Then there is a quantum algorithm that outputs a set $C'$ such that $C\subseteq C'$ with bounded error, and has worst-case
query complexity $\widetilde{O}\left(\frac{R(n)^{2}n}{g(n)^{3/2}}\right)$. 
\end{theorem}

We can assume without loss of generality that the $C$ in \cref{thm:small_cut} is a minimal $st$-cut. While we are not guaranteed that the set $C'$ output by the algorithm referred to in \cref{thm:small_cut} is minimal, it is still an $st$-cut as long as it contains $C$, since its removal will disconnect $s$ and $t$.

To prove \cref{thm:small_cut}, we will use the following variation of the well-known ``coupon collector'' problem.
\begin{lemma}\label{lem:coupon}
Consider repeatedly sampling a random variable $Z$ on a finite set ${\cal S}$. Let $C\subseteq {\cal S}$ be such that for each $e\in C$, $\Pr[Z=e]\geq B$. 
Let $T$ be the number of samples to $Z$ before we have sampled each element of $C$ at least once. Then $\mathbb{E}[T]=O\left(\frac{\log |C|}{B}\right)$.
\end{lemma}
\begin{proof}
For $i\in\{1,\dots,|C|\}$, the probability that $Z$ is a new element of $C$, after $i-1$ elements have already been collected, is 
$p_i \geq (|C|-(i-1))B$. Let $T_i$ be the number of samples to $Z$ after sampling $(i-1)$ elements of $C$, until we sample $i$ elements of $C$, so $T_i$ is a geometric random variable with 
\begin{equation}
\mathbb{E}[T_i] = 1/p_i \leq ((|C|-(i-1))B)^{-1}.
\end{equation}
From this we can compute
\begin{equation}
\mathbb{E}[T] = \sum_{i=1}^{|C|}\mathbb{E}[T_i] \leq \sum_{i=1}^{|C|}\frac{1}{(|C|-(i-1))B} = \frac{1}{B}\sum_{j=1}^{|C|}\frac{1}{j} = \Theta\left(\frac{\log |C|}{B}\right).\qedhere
\end{equation}
\end{proof}

\begin{proof}[Proof of \cref{thm:small_cut}]
We use parameters $T'$ and $\epsilon$, to be defined shortly, and $\delta=1/4$. 
Our strategy is to repeatedly run \texttt{WitnessGeneration}$(\mathcal{P}_{G_{st}},O_x,\epsilon,\delta)$ (\cref{alg:state_generation}) to produce
an approximate witness state, and then measure the 
resultant state in the standard basis to get an edge $e$, which we add to $C'$. We repeat this $T'$ times, before outputting $C'$.

Let $Z$ be the random variable on ${E}\cup\{\text{Failure}\}$ representing the measured output of one call to \cref{alg:state_generation}.
We set ${\epsilon}=\Theta\left(\frac{g(n)}{R(n)}\right)$ small enough so that if the algorithm does not fail,
we produce a state
$\ket{\theta^*}/\|\ket{\theta^*}\|+\ket{\eta}$ 
where $\|\ket{\eta}\|^2\leq g(n)/R(n)$ (see \cref{eq:created_state} and following discussion). Then the probability that
we sample an edge $e'\in C$ when we measure in the standard basis is
\begin{align}
\left\|\bra{e'}\left(\ket{\theta^*}/\|\ket{\theta^*}\|+\ket{\eta}\right)\right\|^2 
&=\|2\theta^*(e')/\sqrt{R_{s,t}(G(x))}-\braket{e'}{\eta}\|^2\nonumber\\
&\geq\|2\sqrt{g(n)/R(n)}- \sqrt{g(n)/R(n)}\|^2\nonumber\\
&=\Omega(g(n)/R(n)).
\end{align}
Since the probability of one call to \cref{alg:state_generation} not failing is $1-\delta=\Omega(1)$, for every $e'\in C$, we have $\Pr[Z=e']\geq B$ for some $B=\Omega(g(n)/R(n))$. Thus, by \cref{lem:coupon}, the expected number of calls to \cref{alg:state_generation} before $C\subseteq C'$ is at most:
\begin{equation}
\mathbb{E}[T] = O\left(\frac{R(n)}{g(n)}\log |C|\right) = O\left(\frac{R(n)}{g(n)}\log n\right).
\end{equation}
By Markov's inequality, if we set $T' = 100\mathbb{E}[T]$, the algorithm will succeed with bounded error.

By \cref{thm:witness_generation}, each call to \cref{alg:state_generation} has 
expected query complexity
\begin{equation}
\widetilde{O}\left(\sqrt{\frac{R_{s,t}(G(x))n^2}{\epsilon}}\right)
=\widetilde{O}\left(n\sqrt{\frac{R(n)}{g(n)/R(n)}}\right)=\widetilde{O}\left(\frac{n R(n)}{\sqrt{g(n)}}\right),
\end{equation}
so the total expected query complexity is
\begin{equation}
\widetilde{O}\left(T'\frac{n R(n)}{\sqrt{g(n)}}\right) = \widetilde{O}\left(\frac{R(n)}{g(n)} \frac{n R(n)}{\sqrt{g(n)}}\right)= \widetilde{O}\left( \frac{n R(n)^2}{g(n)^{3/2}}\right).
\end{equation}
We can get a worst case algorithm by stopping after 100 times the expected number of steps, if the algorithm is still running, and outputting the current $C'$. We have no guarantee on the correctness of $C'$ in that case, but by Markov's inequality, this only happens with probability $1/100$.
\end{proof}

We can use \cref{thm:small_cut} to prove the following result for finding
an $st$-cut set in a particular family of graphs with expander subgraphs and
a single $st$-cut edge.
\begin{restatable}{corollary}{bridge}\label{thm:bridge}
Let $G=(V,E)$ with $s,t\in V$ be a family of $n$-vertex graphs, and suppose we are additionally
promised that $G(x)$ consists of two
disjoint, $d$-regular (for $d\geq 3$), constant expansion subgraphs, each on
$n/2$ vertices, where $s$ and $t$ are always put in separate subgraphs, plus a
single additional edge connecting the two subgraphs. Then there is a quantum algorithm that finds the $st$-cut edge with bounded error in worst-case
$\widetilde{O}(n)$ queries, while any classical algorithm has query complexity $\Omega(n^2)$.
\end{restatable}

\begin{proof}
For a classical algorithm, even if the algorithm had complete knowledge
of the two subgraphs, there would be $\Omega(n^2)$ possible locations for the 
connecting edge, reducing the problem to search, requiring $\Omega(n^2)$ queries.

For the quantum algorithm, note that the maximum effective
resistance between any two points in a $d$-regular (for $d\geq 3$), constant
expansion graph on $n$-vertices is $O(1/d)$ \cite
{chandra1996electrical}. Thus $R_{s,t}(G(x))=\Omega(1)$. Additionally, since there is
only one edge $e'$ connecting the two subgraphs, the optimal unit $st$-flow on $e'$, 
$\theta^*(e')$, must be equal to $1$. 

Applying \cref{thm:small_cut} with $R(n)=O(1)$ and $g(n)=\Omega(1)$, we get a worst-case bounded error quantum query complexity $\widetilde{O}(n).$
\end{proof}


\subsection{Path Finding}\label{sec:path-finding}
In this section, we consider the problem of finding an $st$-path in 
$G(x)$, which we denote
$st\mbox{-}\textsc{path}_G(x)$. That is, given query access to a string $x$ that determines a subgraph $G(x)=(V,E(x))$ of an $n$-vertex graph $G$, as described in \cref{sec:graph-theory} (if $G$ is a complete graph, $x$ is just the adjacency matrix of $G(x)$),  with $s,t\in V$ such that there is at least one path from $s$ to $t$ in $G(x)$, output a path from $s$ to $t$. A
path is a sequence of \emph{distinct} vertices $\vec{u}=(u_0,\dots,u_{\ell})$ such that $s=u_0$, $t=u_{\ell}$, and for all $i\in [\ell]$, $(u_{i-1},u_i)\in \overrightarrow{E}(G(x))$. 

To solve $st\mbox{-}\textsc{path}_G$, one might expect that we could simply apply
\cref{alg:edge_finder} multiple times, storing each edge's endpoints and identifying vertices
 of the endpoints of found edges to reduce the size of the graph, until a path is found. However, such an
 algorithm could run into challenges that could produce slow running times. For example, 
 in a graph where there are many $st$-paths, the algorithm could
 spend too much time sampling edges from different paths, rather than
 focusing on completing a single path. In the case of a single $st$-path,
 such a strategy would not take advantage of the fact that once one edge on
 the path is found, the problem reduces to two connectivity subproblems
 (from $s$ to the found edge, and from $t$ to the found edge) that each typically have significantly
 smaller query complexities than the original problem. 

Thus we develop two algorithms that allow us to prove tighter expected query complexity
bounds than Ref.~\cite{durr2006quantum} for the case of short longest
$st$-paths, one in the case of a single $st$-path, and one for generic graphs.

Before getting into quantum algorithms for path detection, we note the
following corollary of \cref{thm:class_edge_finding}, via a reduction to
path finding from path-edge finding, that characterizes the classical query
complexity of path finding in the case of short longest $st$-paths:
\begin{corollary}
Let $G=(V,E)$ with $s,t\in V$ be an $n$-vertex complete graph and suppose we are promised that $G(x)$ has a path of length $L$ for $L\in[3,n/4]$ between $s$
 and $t$. Then $st\mbox{-}\textsc{path}_G(x)$
has randomized
 query complexity $\Omega(n^2)$.
\end{corollary}

\subsubsection{Graph with a Single Path}

When the graph  $G(x)$ is known to have a single $st$-path, we will we
use a divide-and-conquer algorithm to find the path. To show that
the divide-and-conquer approach is useful, we first consider the simpler algorithm 
(as described above) that
 uses \cref{thm:edge_finder} to find an edge 
$\{u,v\}$ on the path, and then once that edge is found, the algorithm is run on a new
graph where vertices $u$ and $v$ are identified. This process is continued
until the edge $\{s,t\}$ is found. Thus if the length of the path is
initially $L$, after an edge is found, the path length will be $L-1$, and
then $L-2$ in the next iteration, etc. Ignoring error, and assuming the
algorithm finds an edge in each round, by \cref{thm:edge_finder}, the query
complexity at the $i$th round will be $\widetilde{O}(n\sqrt{L-i}).$ Over the
course of the $L$ rounds, the total query complexity will be
\begin{equation}
\sum_{i=0}^{L-1}\widetilde{O}(n\sqrt{L-i})=\widetilde{O}\left(nL^{3/2}\right).\label{eq:bad_approach}
\end{equation}
For $L\geq n^{2/3}$, this algorithm does not even outperform the best classical 
algorithm, and for $L\geq n^{1/3}$ it does not outperform the quantum algorithm of 
Ref.~\cite{durr2006quantum}.

We instead consider the following divide-and-conquer approach, described in detail
in \cref{alg:SinglePathFind}. We
use \cref{alg:edge_finder} to find a set of edges, some of which are
very likely to be on the path. Then we use \cref{cor:st-conn} to verify which
of those edges is actually on the path, and \cref{thm:witness-est} to ensure we choose an edge near the center of the path, so we are left with two subproblems of approximately half the size. Finally we make two recursive calls to
find the unique path from $s$ to the found edge, and the unique
path from $t$ to the found edge.

\begin{algorithm}
    \DontPrintSemicolon
    \SetKwInOut{Input}{Input}
    \SetKwInOut{Output}{Output}
    \SetKwRepeat{Do}{do}{while}
    \Input{Failure tolerance $p>0$, oracle $O_x$ for the graph $G(x)=(V,E(x))$, $s,t\in V$ such that there is a unique path from $s$ to $t$. }
    \Output{With probability $1-O(p)$, a set of edges whose vertices form a path from $s$ to $t$ in $G(x)$.}
    \tcp{Base Cases}
    \lIf{$s=t$}{
      Return $\emptyset$
    }\label{line:bc1}
    \lIf{$\{s,t\}\in E(x)$}{
      Return $\{s,t\}$
    }\label{line:bc2}
    \tcp{Finding Possible Edges on Path}
   $\eps_1\gets \frac{1}{\log n}$ \tcp{Any $\eps_1=o(1)$ that is inverse polylog$(n)$ would suffice}
   $ S\gets \emptyset$\;
   $\ell\gets\frac{2\log(n^5/p)}{\eps_1}$
    \For{$i=1 $ \KwTo $\ell$}
    {
       $e\gets\texttt{EdgeFinder}(O_x,\eps_1,G,s,t)$ (\cref{alg:edge_finder})\label{line:edgeset} \;
      \lIf{$e\neq$ ``Failure'' and $e=(u,v)\in \overrightarrow{E}(x)$}
      {
        $S=S\cup \{(u,v),(v,u)\}$
      }
    }
    \tcp{Finding a possible edge that is actually on a path}
    $\delta\gets p/(\ell n^5)$\;
    \For{$(u,v)\in S$}{
      Initialize \texttt{PathDetection}$(O_x,G^-_{\{u,v\}},s,u,\delta)$ (\cref{cor:st-conn})\;
      Initialize \texttt{PathDetection}$(O_x,G^-_{\{u,v\}},v,t,\delta)$ 
    }
    $flag\gets$ True\;
    \While{flag}
    {
      Run in parallel each \texttt{PathDetection} algorithm initialized in the prior \textbf{for} loop, until each algorithm applies $O_x$ once or terminates (or do nothing for those algorithms that have terminated previously)\;
      \For{$(u,v)\in S$}{
         \If{{\em{\texttt{PathDetection}}}$(O_x,G^-_{\{u,v\}},s,u,\delta)$ and {\em{\texttt{PathDetection}}}$(O_x,G^-_{\{u,v\}},v,t,\delta)$ have both terminated in this iteration of the \textbf{while} loop and both detected paths \label{line:twoPaths}} 
         {
	$\eps_2\gets \sqrt{\eps_1}$, $\eps_3\gets 2\sqrt{\eps_1}$\;
         $\tilde{k}\gets$\texttt{WitnessSizeEst}$(O_x,G,s,u,\eps_2,\delta)$ (\cref{thm:witness-est}) \tcp{estimate of dist.~$s$ to $u$}\;
         \If{$|\tilde{k}-L/2|\leq \eps_3 L$ \label{line:condWitEst}}{ 
	          $(u^*,v^*)\gets (u,v)$\;
        	  $flag\gets$ False\;
		  }
          }
      
    }
    }
    \tcp{Recursive call}
    Return $\{(u^*,v^*)\} \cup \texttt{SinglePathFinder}(O_x,p,G,s,u^*)
    \cup \texttt{SinglePathFinder}(O_x,p,G,v^*,t)$ \label{line:recursive}\;

    \caption{\texttt{SinglePathFinder}$(O_x,p,G,s,t)$}
    \label{alg:SinglePathFind}
\end{algorithm}

\begin{restatable}{theorem}{singlepath}\label{thm:singlePath} Let $p\geq 0$,
 and $G=(V,E)$ with $s,t\in V$ be a family of $n$-vertex graphs, and suppose we are promised that $G(x)$ contains a single $st$-path
 of some length $L$ ($L$ may depend on $x$ and need not be known ahead of time). Then there is a quantum algorithm (\cref{alg:SinglePathFind}) that with probability $1-O(p)$ solves $st\mbox{-}\textsc{path}_G(x)$ and
 uses $\widetilde{O}(nL^{1+o(1)}\log^2(1/p))$ expected queries on input~$x$.
\end{restatable}

\begin{proof}
We first analyze the probability of error, then we prove the correctness of \cref{alg:SinglePathFind}, assuming that no
errors are made, and finally, we analyze the query complexity.

We will stop the algorithm after $O(n)$ recursive calls. Since each recursive call returns an
edge, and any path has length at most $n$, this termination will not affect the success
probability. We then bound our probability of error by $O(p/n^4)=O(p)$, by showing that the
failure probability in each recursive call is $O(p/n^5)$.

\begin{samepage}
We say a failure occurs (in some recursive call) if any of the following happens:
\begin{enumerate}
\item Any one of the at most $4\ell$ \texttt{PathDetection}
algorithms errs. This has probability $O(\ell\delta)=O(p/n^5)$, by our choice of $\delta=p/(\ell n^5)$. 
\item One of the at most $O(\ell)$ calls to \texttt{WitnessSizeEst} produces an estimate that is not within the desired relative error. This has probability $O(\ell\delta)=O(p/n^5)$. 
\item None of the $\ell$ iterations of \texttt{EdgeFinder} produces an edge that is on the $st$-path, and moreover, that is within $(\eps_3-\eps_2)L=\sqrt{\eps_1}L$ of the middle of the path. The absence of this type of failure is sufficient to guarantee that the condition on \cref{line:condWitEst} will be satisfied, as long as \texttt{WitnessSizeEst} is also successful.  
\label{item:failure3}
\end{enumerate}
\end{samepage}
We analyze the probability of the last event, assuming the first two do not occur. 
Let $e_0,\dots,e_{L-1}$  denote the path edges, in order, in the unique $st$-path in $G(x)$.
For one of the $\ell$ runs of \texttt{EdgeFinder}, the probability that 
it does not output ``Failure'' is $\eps_1$. Conditioned on the output of \texttt{EdgeFinder} not being ``Failure,'' by \cref{thm:edge_finder}, we sample from a distribution $\hat q$ that is $\sqrt{\eps_1}$-close in total variation distance to the uniform distribution over edges on the $st$-path. Thus, the probability that we sample an edge in the set 
\begin{equation}
R=\{e_k:k\in [L/2-(\eps_3-\eps_2)L,L/2+(\eps_3-\eps_2)L]\},\label{eq:R-path-edges}
\end{equation}
where $e_k$ is the $k^{\textrm{th}}$ path edge, is:
\begin{equation}
\hat{q}(R) \geq \frac{|R|}{L} - \sqrt{\eps_1} = 2(\eps_3-\eps_2) - \sqrt{\eps_1} = 2(2\sqrt{\eps_1}-\sqrt{\eps_1})-\sqrt{\eps_1}=\sqrt{\eps_1}.
\end{equation}
Thus, using $\eps_1\leq 1/2$, each of the $\ell$ samples has probability at least $(1-\eps_1)\sqrt{\eps_1}\geq \sqrt{\eps_1}/2$ of being a path edge in the correct range, $R$. 
Using
Hoeffding's bound, the probability that none of them is a path edge in the correct range is thus at most:
\begin{equation}
e^{-2\ell(\sqrt{\eps_1}/2)^2}=e^{-\ell\eps_1/2}=e^{-\log(n^5/p)}=O(n^{-5}p) 
\end{equation}
by our choice of $\ell=2\log(n^5/p)/\eps_1$. 
The total probability of failure in one round is thus at most $O(p/n^5)$.

We prove correctness using induction on $L$, the length of the path, assuming no failure occurs.
For the base case, if $L=0$ or $L=1$, we will correctly return the path in \cref{line:bc1,line:bc2}. 

For the inductive case, let $L'\geq 1$. We assume \texttt{SinglePathFinder} works 
correctly for all  lengths $L$ such that  $0\leq L\leq L'$. Now consider a graph 
with $L=L'+1.$ Then assuming no failure, we will sample at least one edge $(u,v)$ in the set $R=\{e_k:k\in [L/2-(\eps_3-\eps_2)L,L/2+(\eps_3-\eps_2)L]\}$ (not doing so is a failure of the type specified by \cref{item:failure3} in the list above). Then if
there are no errors in the \texttt{PathDetection} algorithms, \cref
{line:twoPaths} will be satisfied when $(u,v)$ corresponds to an edge in
the path where $u$ is closer to $s$ and $v$ is closer to $t$. This is
because we have removed $\{u,v\}$ from the graph when we are
running \texttt{PathDetection}, and since there is a unique $st$-path, there
will only be a path from $s$ to $u$ and not from $s$ to $v$, and likewise
for $t$.

Then for every edge $(u,v)$ that we have correctly found using \texttt{PathDetection}
to be on a path, we apply \texttt{WitnessSizeEst} (see \cref{thm:witness-est}) to estimate $R_{s,u}(G(x))$.
If $(u,v)=e_k$, then $e_0,\dots,e_{k-1}$ is the unique $su$-path in $G$, and it
has length $k$, and so $R_{s,u}(G(x))=k$, and thus
\texttt{WitnessSizeEst} is actually estimating $k$. Assuming $(u,v)\in R$, (and
we know this holds for at least one such edge), we have $|k-L/2|\leq (\eps_3-\eps_2)L$. Then 
since we assume \texttt{WitnessSizeEst} does not fail, it outputs an estimate $\tilde{k}$
of $k$, such that $|\tilde{k}-k|\leq \eps_2 k \leq \eps_2 L$. Together, these conditioned imply $|\tilde{k}-L/2|\leq \eps_3 L$, 
which will trigger the \textbf{while} loop to halt. It is possible that we will
break out of the loop for an edge not in $R$, but at the least we know that if no
failure occurs, we will will certainty break out of the \textbf{while} loop with an edge $(u^*,v^*)$
on the path.

Now that we have the edge $(u^*,v^*)$, to find the rest of the path, we just need 
to find the rest of the path from $s$ to $u^*$ and from $v^*$ to $t$. But both of 
these problems will have path lengths between $0$ and $L'$, so by inductive 
assumption, the recursive calls in \cref{line:recursive} will be correct, and will
return the edges on the paths.

Turning to our analysis of the expected query complexity, we first bound the contribution to the expected query complexity in the case of a failure. As just discussed, a failure occurs with probability $O(p/n^4)$. Even in case of failure, each of our $O(n\log(n/p))=O(n^2\log(1/p))$ calls to \texttt{EdgeFinder}, \texttt{PathDetection}, and \texttt{WitnessSizeEst} still has expected query complexity at most $\widetilde{O}(n^{1.5}(1/\eps_1+1/\eps_2^{3/2})\log(1/\delta))=O(n^2\log(1/p))$ (for \emph{any} $x$), for a total query cost of $O(n^4\log^2(1/p))$. 
Thus, the error case contributes an additive
$O(p\log^2(1/p))=O(1)$ to the expected query complexity. 

Next, we create a recurrence relation for the expected query complexity, assuming no failure occurs. Let $\mathbb{E}[T_L]$ be the expected query complexity 
of \cref{alg:SinglePathFind} on a graph with $n$ vertices, when
there is a single path, and that path has length $L$.  
For $k\in \{0,\dots,L-1\}$, let $\tilde{q}_{L}(k)$ be the probability that the path edge that we find, $(u^*,v^*)$, is $e_k$. 
Because we assume no subroutine call fails, we can assume that $\tilde{k}$ is an estimate of $k$ with relative error $\eps_2$, so $|\tilde{k}-k|\leq \eps_2 k\leq \eps_2 L$. From the conditional statement in \cref{line:condWitEst}, we also have $|\tilde{k}-L/2|\leq  \eps_3L$. Taken together, these imply:
\begin{equation}
\abs{k-L/2}\leq (\eps_2+\eps_3)L=(\sqrt{\eps_1}+2\sqrt{\eps_1})L=3\sqrt{\eps_1}L.
\end{equation}
Thus with certainty (assuming no failure occurs), we will exit the 
\textbf{while} loop with $(u^*,v^*)=e_k$, for $k\in [(1/2-3\sqrt{\eps_1})L,(1/2+3\sqrt{\eps_1})L]$, so:
\begin{multline}\label{eq:divide_setup_1}
\mathbb{E}[T_L]=\widetilde{O}(\ell n\sqrt{L}/\eps_1)+\widetilde{O}(\ell n\sqrt{L}\log(1/\delta))+\widetilde{O}\left(\ell\frac{n\sqrt{L}}{\eps_2^{3/2}}\log(1/\delta)\right)\\
+\sum_{k=\lceil (1/2-3\sqrt{\eps_1})L \rceil}^{\lfloor (1/2+3\sqrt{\eps_1})L \rfloor}\tilde{q}_{L}(k)\left(\mathbb{E}[T_{k}]+\mathbb{E}[T_{L-k-1}]\right),
\end{multline} 
where the first three terms come from: (1) running \texttt
 {EdgeFinder} (\cref{alg:edge_finder}, \cref{thm:edge_finder}) $\ell$ times; (2) at most $O(\ell)$ parallel \texttt
 {PathDetection} (\cref{cor:st-conn}) algorithms; and (3) running \texttt{WitnessSizeEst} (\cref{thm:witness-est}) $O(\ell)$ times;
 and the final term from the two recursive calls.
 
To get a function that is strictly increasing in $L$, let $T_L'\coloneqq\max_{k\leq L}\mathbb{E}[T_k]$, 
so in particular $\mathbb{E}[T_L]\leq T_L'$, and $T_L'$ also satisfies the recursion in \cref{eq:divide_setup_1} (with $=$ replaced by $\leq$).
Then we have, for any $k\in [(1/2-3\sqrt{\eps_1})L,(1/2+3\sqrt{\eps_1})L]$,
\begin{equation}
\mathbb{E}[T_k]+\mathbb{E}[T_{L-k-1}]\leq 2 T_{(1/2+3\sqrt{\eps_1})L}'.
\end{equation}
Thus, continuing from \cref{eq:divide_setup_1}, and also using $1/\eps_1=\log n$ and $1/\eps_2=1/\sqrt{\eps_1}=\sqrt{\log n}$, $\ell=2\log(n^5/p)/\eps_1=O(\log(1/p)\log^2n)$, and $\log(1/\delta)=O(\log(\ell n/p))=\log(1/p)\text{polylog}(n,\log(1/p))$,
we get
\begin{equation}\label{eq:divide_setup}
\mathbb{E}[T_L] \leq T_L'\leq \widetilde{O}\left(n\sqrt{L}\log^2({1}/{p})\right)+2T_{(1/2+3\sqrt{\eps_1})L}.
\end{equation} 

To analyze this recurrence, we add up the number of queries made in every recursive call. At the $i^{\mathrm{th}}$ level of recursion, there are $2^i$ recursive calls, and each one makes $\widetilde{O}\left(n\sqrt{L/b^i}\log^2(1/p)\right)$ queries itself, where $b=(1/2+3\sqrt{\eps_1})^{-1}$, before recursing further.
Thus
\begin{equation}\label{eq:divide_recurrence}
\begin{split}
\mathbb{E}[T_L] &\leq \widetilde{O}\left(n\sqrt{L}\log^2({1}/{p})\right)+\sum_{i=1}^{\log_bL}2^i\sqrt{\frac{L}{b^i}}\cdot \widetilde{O}\left(n\log^2(1/p)\right)\\
&\leq \widetilde{O}\left(n\sqrt{L}\log^2({1}/{p})\right)+\widetilde{O}\left(n\sqrt{L}\log^2({1}/{p})\right)\left(2/\sqrt{b}\right)^{\log_b L}.
\end{split}
\end{equation} 
Letting $\eta \coloneqq \frac{1}{1+\frac{1}{6\sqrt{\eps_1}}}=O(1/\sqrt{\log n})$ since $\eps_1=1/\log n$, so that $b=2(1-\eta)$, we have:
\begin{equation}
\begin{split}
\log\left(2/\sqrt{b}\right)^{\log_b L} &= \left(1 - \frac{1}{2}\log b\right)\frac{\log L}{\log b}
= \left(\frac{1}{\log b} - \frac{1}{2}\right)\log L\\ 
&= \left(\frac{1}{1-\log\frac{1}{1-\eta}}-\frac{1}{2}\right)\log L
= \left(\frac{1}{2} + \frac{\log\frac{1}{1-\eta}}{1-\log\frac{1}{1-\eta}}\right)\log L\\
\mbox{so }\left(2/\sqrt{b}\right)^{\log_b L} &= L^{\frac{1}{2}+o(1)},
\end{split}
\end{equation}
where we used $\frac{\log\frac{1}{1-\eta}}{1-\log\frac{1}{1-\eta}}=o(1)$, since $\log\frac{1}{1-\eta}=o(1)$, which follows from $\eta=o(1)$. Thus, continuing from \cref{eq:divide_recurrence}, we have:
\begin{equation}
\mathbb{E}[T_L] = \widetilde{O}\left(n\sqrt{L}\log^2({1}/{p})\right)L^{\frac{1}{2}+o(1)}
= \widetilde{O}\left(n{L}^{1+o(1)}\log^2({1}/{p})\right). \qedhere
\end{equation}
\end{proof}

We note that while our approach in 
\cref{thm:singlePath} outperforms the simpler, non-divide-and-conquer algorithm
analyzed in \cref{eq:bad_approach}, it performs worse than the algorithm of
Ref.~\cite{durr2006quantum} for graphs with $L=\Omega (n^{1/2-o(1)})$. Thus, one could run
\cref{alg:SinglePathFind} until $O\left(n^{3/2}\right)$ queries had been made, and if a
path had not yet been found, switch to the algorithm of Ref.~\cite{durr2006quantum}.


\subsubsection{Path Finding in Arbitrary Graphs}

For the more general case of $st\mbox{-}\textsc{path}_G(x)$ when $G(x)$ is not known 
to only have one $st$-path, while it is possible that an algorithm similar
to \cref{alg:SinglePathFind} would work, we have not been able to bound the
running time effectively. This is because in the case of a single path, once you
find an intermediate edge on the path, the longest paths from $s$ and $t$ to
that edge must be shorter than the length of the longest path from $s$ to
$t$. This ensures that subproblems take shorter time than the original
problem. With multiple paths, we no longer have that guarantee.

However, we provide an alternative approach that, while not as fast as \cref
{alg:SinglePathFind}, still provides an improvement over the algorithm
of \cite{durr2006quantum} for graphs in which all (self-avoiding) paths from $s$ to $t$ are short. Our approach
does not make use of our path-edge sampling algorithm as a subroutine, and instead uses the
path detection algorithm of \cref{cor:st-conn} to decide whether there are paths through various
subgraphs, and then uses that information to find each edge in a path in order from $s$ to $t.$
In this way, we
avoid the problem of subproblems being larger than the original problem,
since if the longest path from $s$ to $t$ has length $L$, and the first edge we find on the path
is $(s,u)$, then
longest path from $u$ to $t$ that doesn't go through $s$ must have length at most $L-1.$ However, we lose the advantage
of a divide-and-conquer approach.

To find the first edge on a path, we use a group testing approach. We divide the
neighbors of $s$ in $G$ into two sets, $S_1$ and $S_2$ and run path detection
algorithms in parallel on two subgraphs of $G(x)$, one with edges from $s$ removed, except
those to vertices in $S_1$ (that is, 
$G^-_{\{\{s,u\}\in E:u\in S_1\}}$), and one with edges from $s$ removed, except those to vertices in $S_2$. We will detect which of these subgraphs contains a path, and we will know there
is a path whose first edge goes from $s$ to a vertex in the corresponding
set ($S_1$ or $S_2$). Then we divide that set into half again, and repeat, until we have
narrowed down our set to one vertex $u$, that must be the first vertex on a
path from $s$ to $t.$ 

At this point we have learned the first edge on a path from $s$ to $t$. We then 
consider $G_s^-$, which is $G$ with vertex $s$ removed, and
recursively iterate this procedure to learn the first edge on a path from $u$ to $t$.

\begin{algorithm}[H]
    \DontPrintSemicolon
    \SetKwInOut{Input}{Input}
    \SetKwInOut{Output}{Output}
    \SetKwRepeat{Do}{do}{while}
    \Input{Failure tolerance $p$, oracle $O_x$ for the graph $G(x)=(V,E(x))$, $s,t\in V$ such that there is a path from $s$ to $t$. }
    \Output{With probability $1-O(p)$, a sequence of edges whose vertices form a path from $s$ to $t$ in $G(x)$}
    $\delta\gets p/(n^4\log n)$\;
    \tcp{Base Case}
    \lIf{$s=t$}{
      Return $\emptyset$\label{line:bc}
    }
    \tcp{Finding the first edge in a path from $s$ to $t$}
    $S,E_s\gets\{\{s,v\}:\{s,v\}\in E\}$\;
    \While{$|S|> 1$}{
      Divide $S$ into two approximately equal, disjoint sets, $S_1$ and $S_2$\;
      Run in parallel the following two algorithms such that the queries implemented by each algorithm stays within $1$ of the other at all times, until one outputs 1 \;
      \begin{itemize}
      \item \texttt{PathDetection}$(O_x,G^-_{E_s\setminus S_1},s,t,\delta)$ (\cref{cor:st-conn}) \label{line:PD1}\;
      \item \texttt{PathDetection}$(O_x,G^-_{E_s\setminus S_2},s,t,\delta)$ (\cref{cor:st-conn}) \label{line:PD2}\;
      \end{itemize}
      \eIf{$G^-_{E_s\setminus S_1}$ algorithm output 1}{
        $S\gets S_1$\;
      }
      {
        $S\gets S_2$\;
      }
      \lIf{neither \texttt{PathDetection} call outputs $1$}{return ``Failure''}

    }
    Return $((s,u))\frown $ \texttt{GeneralPathFinder}$(O_x,G^-_s,u,t,p)$\label{line:rec} \tcp{$\frown$ indicates concatenation of sequences}
    \caption{\texttt{GeneralPathFinder}$(O_x,G,s,t,p)$}
    \label{alg:GeneralPathFinder}
\end{algorithm}

\begin{restatable}{theorem}{generalPathFinder}\label{thm:generalPathFinder} 
Let $p\geq 0$, and $G=(V,E)$ with $s,t\in V$ be a family of $n$-vertex
 graphs, and suppose we are promised that there is a path from $s$ to $t$ in $G(x)$. On input $x$, if the longest $st$-path in $G(x)$ has length $L$ ($L$
 need not be known ahead of time), there is a quantum
 algorithm (\cref{alg:GeneralPathFinder}) that returns the edges on a path with probability $1-O(p)$ and
 uses $\widetilde{O}(nL^{3/2}\log(1/p))$ expected queries.
\end{restatable}

We note that \cref{alg:GeneralPathFinder} performs worse than the algorithm of
Ref.~\cite{durr2006quantum} for graphs with $L>n^{1/3}$. Thus, one could run
this algorithm until $O\left(n^{3/2}\right)$ queries had been made, and if a
path had not yet been found, switch to the algorithm of 
Ref.~\cite{durr2006quantum}.

\begin{proof}
We first analyze the probability of error in \cref{alg:GeneralPathFinder}. Over the course
of the algorithm, there will be $O(n)$ recursive calls (since each recursive
call returns an edge). We bound our probability of error to $O(p/n^3)=O(p)$, by showing that the failure probability in each recursive call is $O(p/n^{4})$.

We consider a recursive call to have an error if any of the $O(\log n)$ calls to \texttt{PathDetection} fails. Because of our choice of $\delta=p/(n^4\log n)$, each call fails with probability $O(p/(n^4\log n)$, so the probability that all such calls succeed is 
\begin{equation}
(1-O(p/(n^4\log n))^{O(\log n)}=1-O(p/n^4),
\end{equation}
so the probability that at least one call fails is $O(p/n^4)$ and the probability that 
any call fails is $O(p/n^3)$.  

Even in the case of a failure, the expected query complexity of the algorithm is at most $O(n^3\log(1/p))$, since at most $O(n\log n)$ calls to \texttt{PathDetection} are made over the course of the algorithm, each of which has expected query complexity $O(n^{3/2}\log(n/p))$ (for \emph{any} $x$). Thus, the overall contribution to the expected query complexity of \cref{alg:GeneralPathFinder} in the error case is at most $O((p/n^3)n^{3}\log(1/p))=O(1)$.

Thus, we can analyze the expected query complexity of \cref{alg:GeneralPathFinder} assuming no errors occur. When the longest path length between
$s$ and $t$ is $L$, then at least one of the pair of \texttt
{PathDetection} subroutines that are run in parallel will have expected query complexity $\widetilde{O}(n\sqrt{L}\log(1/\delta))$. This is because, as long as there is not an
error, the first edge in a path with length at most $L$ must be contained in
either $S_1$ or $S_2$, so there will be a path in one of the two parallel
subroutines, it will halt after $\widetilde{O}(n\sqrt{L}\log(1/\delta))$ expected queries, since, for any $G'$, $R_{s,t}(G')$ is upper bounded by the length of any $st$-path in $G'$. 

Let $\mathbb{E}[T_L]$ be the expected query complexity of \cref{alg:GeneralPathFinder} when all $st$-paths in $G(x)$ have length at most $L$. Then a recurrence relation for the expected query complexity is
\begin{equation}
\mathbb{E}[T_L]=\widetilde{O}(n\sqrt{L}\log(1/p))+\mathbb{E}[T_{L-1}], \qquad T_0=O(1).
\end{equation}
The $\widetilde{O}(n\sqrt{L}\log(1/p))$ comes from the $O(\log n)$ iterations of \texttt{PathDetection}, each of which has expected query complexity at most $\widetilde{O}(n\sqrt{L}\log(1/\delta))=\widetilde{O}(n\sqrt{L}\log(1/p))$.
Solving this recurrence, we find that
\begin{equation}
\mathbb{E}[T_L]=\widetilde{O}(nL^{3/2}\log(1/p)).
\end{equation}

Finally, we prove the correctness of \cref{alg:GeneralPathFinder} using induction
on the length of the longest path from $s$ to $t$ assuming that no errors are
made. 
For the base case, if $L=0$ we will correctly return the path in \cref{line:bc}.

For the inductive case, let $k\geq 0$, and assume \texttt
{GeneralPathFinder} works correctly for all graphs whose longest path length
from $s$ to $t$ is $L$, where $0\leq L\leq k$. Now consider a graph with
$L=k+1.$ Then as long as none of the $2\lceil\log n\rceil$ iterations of \texttt{PathDetection} in \cref
{line:PD1,line:PD2} fail, we will find an edge $\{s,u\}$ on a path from $s$
to $t$. This is because at each iteration of \cref{line:PD1,line:PD2}, we
find a set of vertices that we know contains the second vertex (first 
vertex after $s$) in a path from $s$ to $t$. At each iteration, the number of vertices in the set for which we have this knowledge decreases by a factor of 2, until we have a set with just one vertex, which must be the next vertex in our path
after $s$.

Once we have found the first edge $\{s,u\}$ of the path, we have a new
problem of finding a $ut$-path on a graph with $s$ removed. But because the longest path from $s$ to $t$
was at most $L$, the longest path from $u$ to $t$ that does not go through $s$
must be at most $L-1$, so by our inductive assumption, the recursive call
to \texttt{GeneralPathFinder} in \cref{line:rec}, which finds a $ut$-path on
the graph with vertex $s$ removed, will be correct.
\end{proof}


\paragraph{Acknowledgements} We thank Jana Sot\'akov\'a and Mehrdad Tahmasbi
 for insightful discussions about path finding via edge sampling. SK and SJ
 were sponsored by the Army Research Office and this work was accomplished
 under Grant Number W911NF-20-1-0327. The views and conclusions contained in
 this document are those of the authors and should not be interpreted as
 representing the official policies, either expressed or implied, of the Army
 Research Office or the U.S. Government. The U.S. Government is authorized to
 reproduce and distribute reprints for Government purposes notwithstanding
 any copyright notation herein. SJ is supported by NWO Klein project number
 OCENW.Klein.061, and the European Union (ERC, ASC-Q, 101040624). Views and
 opinions expressed are however those of the authors only and do not
 necessarily reflect those of the European Union or the European Research
 Council. Neither the European Union nor the granting authority can be held
 responsible for them. SJ is a CIFAR Fellow in the Quantum Information
 Science Program.

\bibliography{spanPrograms1}
\bibliographystyle{alpha}

\appendix

\section{Proofs of Flow Properties}
\label{app:Flow}

In this appendix, we list and prove several results about flows on graphs. We first restate and prove \cref{lem:flow-paths}, which tells us that the optimal flow state for a graph $G(x)$, which is the positive witness for $x$ in ${\cal P}_{G_{s,t}}$, and thus the state $\ket{\theta^*}/\norm{\ket{\theta^*}}$ approximated in Step 2 of \cref{alg:edge_finder}, is supported only on edges that are on (self-avoiding) $st$-paths. 

Then in \cref{app:flow-gen}, we mention some interpretations of the flow in terms of graph theoretic quantities 
that help provide intuition
for the distribution of edges we get when we measure the flow state. In \cref{app:flow-series-parallel} we give further such interpretations that are specific to the case of series-parallel graphs.

\flowpath*

\begin{proof}
For any $(u,v)\in\overrightarrow{E}=\overrightarrow{E}(G)$, define 
\begin{equation}\label{eq:e-uv}
\ket{e_{u,v}} \coloneqq \frac{1}{\sqrt{2}}\left(\ket{u,v}-\ket{v,u}\right) = -\ket{e_{v,u}},
\end{equation}
so that, in particular, 
$
\ket{\rho_{\vec{u}}} = \sum_{i=0}^{\ell-1}\ket{e_{u_i,u_{i+1}}}
$.
Let 
\begin{equation}
H^-(x)\coloneqq \mathrm{span}\{\ket{e_{u,v}}:(u,v)\in \overrightarrow{E}(G(x))\}
\end{equation}
and note that $\{\ket{e_{u,v}}:(u,v)\in\overrightarrow{E}(G(x)), u<v\}$ is an orthonormal basis for $H^-(x)$. 

We first argue that $\ket{\theta^*}$ is in $H^-(x)$. To see this, note that
\begin{equation}
H^-(x)^\bot\cap H(x) = \mathrm{span}\{\ket{u,v}+\ket{v,u} : (u,v)\in\overrightarrow{E}(G(x))\}\subseteq \ker(A)\cap H(x).
\end{equation}
An optimal witness $\ket{\theta^*}$ must be orthogonal to $\ker(A)\cap H(x)$ (or we could make it smaller while keeping it a witness by subtracting its projection onto $\ker(A)\cap H(x)$), so 
$
\ket{\theta^*}\in H^-(x)
$ 
meaning that for $(u,v)\in \overrightarrow{E}(G(x))$, $\theta^*(u,v)=-\theta^*(v,u)$. From this we can see that:
\begin{equation}
\ket{\theta^*}= \frac{1}{\sqrt{2}}\sum_{(u,v)\in \overrightarrow{E}(G(x)):u<v}\theta^*(u,v)\ket{e_{u,v}}.
\end{equation}

Note that 
\begin{equation}
A\ket{e_{u,v}} = \frac{1}{\sqrt{2}}\left(\ket{u}-\ket{v} - (\ket{v} - \ket{u})\right) = \sqrt{2}(\ket{u}-\ket{v}).
\end{equation}
so we can express $A\Pi_{H^-(x)}$ as 
\begin{equation}
A\Pi_{H^-(x)}=A\sum_{(u,v)\in \overrightarrow{E}(G(x)):u<v}\ket{e_{u,v}}\bra{e_{u,v}}
=\sqrt{2}\sum_{(u,v)\in \overrightarrow{E}(G(x)):u<v}(\ket{u}-\ket{v})\bra{e_{u,v}}.
\end{equation}
Then we can compute, for $u\in V$:
\begin{equation}
\begin{split}
\bra{u}A\Pi_{H^-(x)} &= 
\bra{u}\sum_{(u',v)\in \overrightarrow{E}(G(x)):u'<v}\sqrt{2}(\ket{u'}-\ket{v})\bra{e_{u',v}}\\
&= \sqrt{2}\sum_{v:(u,v)\in\overrightarrow{E}(G(x)),u<v}\bra{e_{u,v}}-\sqrt{2}\sum_{v:(v,u)\in\overrightarrow{E}(G(x)),v<u}\bra{e_{u,v}}
= \sqrt{2}\sum_{v:(u,v)\in \overrightarrow{E}(G(x))}\bra{e_{u,v}},
\end{split}\label{eq:uA}
\end{equation}
since $\ket{e_{u,v}}=-\ket{e_{v,u}}$.
Since $\ket{\theta^*}$ is a positive witness, we have $A\Pi_{H^-(x)}\ket{\theta^*}=A\ket{\theta^*} = \ket{s}-\ket{t}$, from which we derive the following constraints, which equivalently say that $\theta^*$ must be a unit $st$-flow:
\begin{equation}\label{eq:flow-constraints}
\forall u\in V,\; \sum_{v:(u,v)\in \overrightarrow{E}(G(x))}\braket{e_{u,v}}{\theta^*}
=\left\{\begin{array}{ll}
1/\sqrt{2} & \mbox{if }u=s\\
-1/\sqrt{2} & \mbox{if }u=t\\
0 & \mbox{otherwise.}
\end{array}\right.
\end{equation}

For $\ket{\psi}\in H^-(x)$, we say $\ket{\psi}$ is a  
\emph{circulation} if it satsifies the following linear constraints:
\begin{equation}
\forall u\in V,\; \sum_{v:(u,v)\in \overrightarrow{E}(G(x))}\braket{e_{u,v}}{\psi} = 0.\label{eq:circ-constraints}
\end{equation}
The subspace of $H^-(x)$ of such vectors will be denoted ${\cal C}(x)$. Let $\overline{G}(x)$ 
be the graph $G(x)$, but with an edge $\{s,t\}$ added (if it was not already 
present). We implicitly assume that $\{s,t\}\in E(G)$, which is without loss of 
generality, since we can always label this edge in such a way that it is not 
present in any $G(x)$, so $\overline{G}(x)$ is a subgraph of $G$. 
If we replace $G(x)$ with $\overline{G}(x)$ in each constraint in \cref{eq:circ-constraints}, 
we can define a subspace of $H^-(x)+\mathrm{span}\{\ket{e_{s,t}}\}$, which will be denoted 
${\cal C}'(x)$.
Then it follows from \cref{eq:uA} that ${\cal C}(x)\subseteq {\cal C}'(x)\subseteq \ker(A)$.

We now argue that  
\begin{equation}\label{eq:paths-circ}
\ket{\psi}\coloneqq \ket{\theta^*} - \frac{1}{\sqrt{2}}\ket{e_{s,t}} \in {\cal C}'(x).
\end{equation}
For any $u\in V\setminus\{s,t\}$:
\begin{equation}
\begin{split}
\sum_{v:(u,v)\in \overrightarrow{E}(\overline{G}(x))}\braket{e_{u,v}}{\psi} &= \sum_{v:(u,v)\in \overrightarrow{E}(G(x))}\braket{e_{u,v}}{\theta^*}-\frac{1}{\sqrt{2}}\sum_{v:(u,v)\in \overrightarrow{E}(\overline{G}(x))}\braket{e_{u,v}}{e_{s,t}}=0-0,
\end{split}
\end{equation}
by \cref{eq:flow-constraints}. For $u=s$,
\begin{equation}
\begin{split}
\sum_{v:(u,v)\in \overrightarrow{E}(\overline{G}(x))}\braket{e_{u,v}}{\psi} &= \sum_{v:(s,v)\in \overrightarrow{E}(G(x))}\braket{e_{s,v}}{\theta^*}-\frac{1}{\sqrt{2}}\braket{e_{s,t}}{e_{s,t}}=\frac{1}{\sqrt{2}}-\frac{1}{\sqrt{2}}=0,
\end{split}
\end{equation}
again by \cref{eq:flow-constraints}, and very similarly for $u=t$. This establishes \cref{eq:paths-circ}.

Let ${\cal P}_{s,t}'(x)$ be the span of all $st$-path states in $\overline{G}(x)$. Then by \cref{eq:paths-circ}, we have:
\begin{equation}
\ket{\theta^*}\in {\cal P}_{s,t}'(x)+{\cal C}'(x).
\end{equation}
Suppose $\{s,t\}\not\in E(G(x))$ (the case where it is is simpler), in which case ${\cal P}_{s,t}'(x) = {\cal P}_{s,t}(x)\oplus\mathrm{span}\{\ket{e_{s,t}}\}$, where ${\cal P}_{s,t}(x)$ is the span of all $st$-path states in $G(x)$. Then
\begin{equation}
\ket{\theta^*}\in {\cal P}_{s,t}(x)\oplus \mathrm{span}\{\ket{e_{s,t}}\}+{\cal C}'(x).
\end{equation}
Since $\{s,t\}\not\in E(G(x))$, $\ket{\theta^*}$ is orthogonal to $\ket{e_{s,t}}$. 
Further, by the optimality of $\ket{\theta^*}$, $\ket{\theta^*}$ is orthogonal 
to $H(x)\cap \ker (A)$, so in particular, it is orthogonal to 
${\cal C}'(x)\subseteq H'(x)\cap \ker(A)\subseteq \mathrm{span}\{\ket{e_{s,t}}\}\oplus(H(x)\cap \ker(A))$. 
It follows that $\ket{\theta^*}\in {\cal P}_{s,t}(x)$. 
\end{proof}

\subsection{Flow on General Graphs}\label{app:flow-gen}

\begin{lemma}[\cite{doyleElectric1984}] 
Let $\theta$ be the optimal unit $st$-flow in $G$. 
For any $u,v\in V$, let $Z_{u,v}$ be the number of times a random walker who starts at $s$ and continues until she reaches $t$ moves from vertex $u$ to vertex $v$. Then $\theta(u,v) = \mathbb{E}[Z_{u,v}]-\mathbb{E}[Z_{v,u}]$. 
\end{lemma}

\begin{lemma}[\cite{lyonsProbabilityOnTrees2017}]\label{lem:flow1}
Let $\theta$ be the optimal unit $st$-flow in $G$. Let ${\cal T}_G$ be the set of 
spanning trees of $G$, and for any $\{u,v\}\in E$, let ${\cal N}_G(u,v)$ be the set 
of spanning trees of $G$ whose unique $st$-path contains the directed edge $(u,v)$. 
Then
$$\theta(u,v)=\frac{|{\cal N}_G(u,v)|-|{\cal N}_G(v,u)|}{|{\cal T}_G|}.$$
\end{lemma}

We have the following interpretation of the quantity $q_{u,v}:=\frac{\theta(u,v)^2}{{\cal R}_{s,t}(G)}$, which is the probability with which our edge finding algorithm, \cref{alg:edge_finder}, samples $\{u,v\}$ (by measuring either $(u,v)$ or $(v,u)$):
\begin{restatable}{lemma}{qinterp}\label{lem:q-interp1}
Suppose $\{s,t\}\not\in E(G)$, and let $\overline{G}$ be the graph $G$ with an 
additional $\{s,t\}$ edge, and $\overline{G}/\{s,t\}$ be $\overline{G}$ with this 
edge contracted -- so it is $G$ with the vertices $s$ and $t$ identified.  Then
$$q_{u,v} = \frac{(|{\cal N}_G(u,v)|-|{\cal N}_G(v,u)|)^2}{|{\cal T}_G|\cdot |{\cal T}_{\overline{G}/\{s,t\}}|}.$$
\end{restatable}

\begin{proof}
For an
edge $e\in E(G)$, we denote by $G/e$ the graph that results from $G$ when the
edge $e$ is contracted.
Note that for any $e\in E(G)$, ${\cal T}_G = {\cal T}_{G/e}\cup {\cal T}_{G\setminus e}$ -- 
that is, ${\cal T}_G$ is the (disjoint) union of: (1) the set of spanning trees 
that contain $e$, which are isomorphic to the spanning trees of $G/e$; and (2) the 
set of spanning trees that do not contain $e$, which are isomorphic to the spanning 
trees of $G\setminus e$. In particular, this implies 
that $|{\cal T}_{\overline{G}/\{s,t\}}| = |{\cal T}_{\overline{G}}| - |{\cal T}_{G}|$, 
since $G = \overline{G}\setminus\{s,t\}$. 
Since $\{s,t\}\in E(\overline{G})$, we have, by \cite[Theorem 6]{kookCombinatorialGreensFunction2011},
\begin{equation}
R_{s,t}(\overline{G}) = \frac{|{\cal T}_{\overline{G}/\setminus\{s,t\}}|}{|{\cal T}_{\overline{G}}|}
= \frac{|{\cal T}_{\overline{G}}|-|{\cal T}_G|}{|{\cal T}_{\overline{G}}|}
= 1 - \frac{|{\cal T}_G|}{|{\cal T}_{\overline{G}}|}.
\end{equation}
Since $\overline{G}$ is a parallel combination of $G$ and an edge $\{s,t\}$, and since conductances in parallel add, we get
\begin{equation}
R_{s,t}(\overline{G})=\frac{1}{\frac{1}{R_{s,t}(G)}+\frac{1}{1}} = 1 - \frac{1}{R_{s,t}(G)+1}.
\end{equation}
This gives us:
\begin{equation}
\begin{split}
R_{s,t}(G)+1 &= \frac{1}{1-R_{s,t}(\overline{G})} = \frac{|{\cal T}_{\overline{G}}|}{|{\cal T}_G|}\\
R_{s,t}(G) &= \frac{|{\cal T}_{\overline{G}}| - |{\cal T}_G|}{|{\cal T}_G|} = \frac{|{\cal T}_{\overline{G}/\{s,t\}}|}{|{\cal T}_G|}.
\end{split}
\end{equation}
Thus, by \cref{lem:flow1}:
\begin{equation}
\begin{split}
q_{u,v} &= \frac{\left(|{\cal N}_G(u,v)|-|{\cal N}_G(v,u)|\right)^2}{|{\cal T}_G|^2}\frac{|{\cal T}_G|}{|{\cal T}_{\overline{G}/\{s,t\}}|},
\end{split}
\end{equation}
from which the lemma statement follows.
\end{proof}

\subsection{Flow on Series-Parallel Graphs}\label{app:flow-series-parallel}

In a series-parallel graph, for every edge $\{u,v\}$, there is a unique direction, $(u,v)$ or $(v,u)$, such that every (self-avoiding) $st$-path that traverses the edge $\{u,v\}$ does so in that direction. We will call this the \emph{$st$-direction} of $\{u,v\}$.
\begin{lemma}\label{lem:flow-series-parallel}
Suppose $G$ is a series-parallel graph.
For any edge $\{u,v\}\in E(G)$ with $st$-direction $(u,v)$, 
$$q_{u,v} = \frac{|{\cal N}_G(u,v)|^2}{|{\cal T}_G|\cdot |{\cal T}_{\overline{G}/\{s,t\}}|} = p_{u,v}\cdot p_{u,v}',$$
where $p_{u,v}$ is the probability that a uniformly sampled spanning tree has $(u,v)$ on its unique $st$-path, and $p_{u,v}'$ is the probability that a randomly sampled two-component spanning forest of $G$ that separates $s$ and $t$ has $(u,v)$ in its unique $st$-cut. 
\end{lemma}
\begin{proof}
Since (self-avoiding) $st$-paths can only use $(u,v)$ and not $(v,u)$, we have ${\cal N}_G(v,u)=\emptyset$. Then it follows from \cref{lem:q-interp1} that 
\begin{equation}
q_{u,v} = \frac{|{\cal N}_G(u,v)|^2}{|{\cal T}_G|\cdot |{\cal T}_{\overline{G}/\{s,t\}}|} = \frac{|{\cal N}_G(u,v)|}{|{\cal T}_G|}\frac{|{\cal N}_G(u,v)|}{ |{\cal T}_{\overline{G}/\{s,t\}}|}.
\end{equation}
Then by definition, $|{\cal N}_G(u,v)|/|{\cal T}_G|$ is the probability that a 
spanning tree randomly sampled from ${\cal T}_G$ has $(u,v)$ on its unique 
$st$-path, $p_{u,v}$. Note that ${\cal T}_{\overline{G}/\{s,t\}}$ is isomorphic to 
the set of spanning trees of $\overline{G}$ that contain $\{s,t\}$, which is 
isomorphic to the set of two-component spanning forests of $G$ in which $s$ and $t$ 
are in separate components. Similarly, ${\cal N}_G(u,v)$ is isomorphic (by removal 
of edge $\{u,v\}$) to the set of two-component spanning forests of $G$ that have $s$ 
and $t$ in separate components, and $(u,v)$ in the unique $st$-cut in its 
complement. Thus $|{\cal N}_G(u,v)|/|{\cal T}_{\overline{G}/\{s,t\}}|$ is $p_{u,v}'$.
\end{proof}

By a simple inductive argument, any series-parallel graph is planar. Given an implicit planar embedding of $G$, we let $G^\dagger$ denote its planar dual, defined as follows.
\begin{definition}[Dual Graph]
For a planar embedded graph $G$, its dual is the graph $G^\dagger$ with vertex set, $V(G^{\dagger})$, defined as precisely the \emph{faces} of $G$ (with respect to its implicit fixed embedding), and an edge $\{f,f'\}$ between a pair of faces $f,f'\in V(G^\dagger)$ if and only if the faces $f$ and $f'$ are separated by an edge $e\in E(G)$, in which case, we write $e=\{f,f'\}^\dagger$. If $e=\{u,v\}$, and $f'$ is the face on the clockwise side of $(u,v)$ (that is, to the right of $(u,v)$ when $(u,v)$ points up), then we write $(u,v)^\dagger = (f,f')$. 
\end{definition}

We will consider a graph $(\overline{G}^\dagger)\setminus\{s',t'\}$ obtained by adding an edge $\{s,t\}$ to $G$ (to obtain $\overline{G}$), taking the dual, and then removing the edge $\{s',t'\}=\{s,t\}^\dagger$. The following figure shows an example of such a graph. We have also given each edge in $G$ an orientation, in order to show the corresponding orientations in the dual. Note that this example graph is planar, but not series-parallel. 
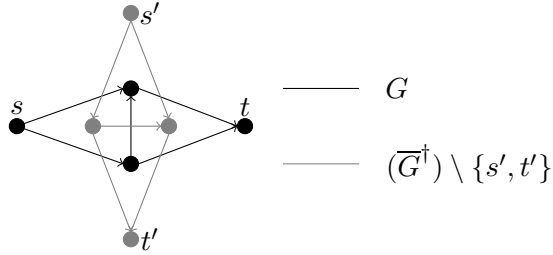
\begin{center}
\begin{tikzpicture}

\filldraw (-1.5,0) circle (.1);		\node at (-1.5,.25) {$s$};
\draw[->] (-1.5,0) -- (-.1,.5);
\draw[->] (-1.5,0) -- (-.1,-.5);
\filldraw (0,.5) circle (.1);
\draw[->] (0,-.5) -- (0,.4);
\filldraw (0,-.5) circle (.1);
\draw[->] (.1,.5) -- (1.4,0);
\draw[->] (.1,-.5) -- (1.4,0);
\filldraw (1.5,0) circle (.1); 		\node at (1.5,.25) {$t$};

\filldraw[gray] (0,-1.5) circle (.1); \node at (.25,-1.5) {$t'$};
\draw[<-,gray] (0,-1.4) -- (.5,-.1);
\draw[<-,gray] (0,-1.4) -- (-.5,-.1);
\filldraw[gray] (.5,0) circle (.1);
\draw[->,gray] (-.5,0) -- (.4,0);
\filldraw[gray] (-.5,0) circle (.1);
\draw[<-,gray] (.5,.1) -- (0,1.4);
\draw[<-,gray] (-.5,.1) -- (0,1.4);
\filldraw[gray] (0,1.5) circle (.1); \node at (.25,1.5) {$s'$};

\draw (2,.5)--(3,.5); \node at (3.5,.5) {$G$};
\draw[gray] (2,-.5)--(3,-.5); \node at (4.45,-.5) {$(\overline{G}^\dagger)\setminus\{s',t'\}$};

\end{tikzpicture}
\end{center}

\begin{lemma}
Suppose $G$ is a series-parallel graph. 
Then $q_{u,v}=p_{u,v} p_{u,v}^{\dagger}$, where $p_{u,v}$ is the probability that a 
uniformly sampled spanning tree of $G$ has $(u,v)$ on its unique $st$-path, and 
$p_{u,v}^\dagger$ is the probability that a uniformly sampled spanning tree 
of $(\overline{G}^\dagger)\setminus\{s',t'\}$ has $(u,v)^\dagger$ on its 
unique $s't'$-path, where $(s',t')=(s,t)^\dagger$. 
\end{lemma}
\begin{proof}
We will make use of the following observation: Fix any planar graph $G$ in which $s$
and $t$ are on the same face. For every $T\in {\cal T}_G$, let $T^\dagger$ be the 
subgraph of $G^\dagger$ such that 
$e\in E(G)\setminus E(T)$ if and only if $e^\dagger\in E(T^\dagger)$.  
Then $T\mapsto T^\dagger$ is a bijection between ${\cal T}_G$ and
${\cal T}_{G^\dagger}$. To see this, note that there is a cycle in $T$ if and only 
if $T^\dagger$ has multiple components -- the subgraph inside the cycle is 
disconnected from the subgraph outside the cycle. And similarly, $T$ has multiple 
connected components if and only if there is a cycle in $T^\dagger$. 
Thus $T^\dagger\in {\cal T}_{G^\dagger}$, and the map is clearly a bijection.

Since $\overline{G}^\dagger\setminus\{s',t'\} = (\overline{G}/\{s,t\})^\dagger$, 
the above observation implies a bijection from ${\cal T}_{\overline{G}/\{s,t\}}$ to ${\cal T}_{(\overline{G}^\dagger)\setminus\{s',t'\}}$, establishing 
\begin{equation}\label{eq:tree-bij}
|{\cal T}_{\overline{G}/\{s,t\}}|=|{\cal T}_{(\overline{G}^\dagger)\setminus\{s',t'\}}|.
\end{equation} 

Next, we will establish 
\begin{equation}
|{\cal N}_G(u,v)|=|{\cal N}_{(\overline{G}^\dagger)\setminus\{s',t'\}}((u,v)^\dagger)|
\end{equation}
by exhibiting a bijection, and, combined with \cref{eq:tree-bij}, the lemma 
statement easily follows from \cref{lem:flow-series-parallel}. First, note that ${\cal N}_{\overline{G}}(u,v)={\cal N}_G(u,v)$, 
since any tree in ${\cal N}_{\overline{G}}(u,v)$ must not contain $\{s,t\}$, or the 
only edge on the $st$-path would be $(s,t)$. The map $\phi(T)=T^\dagger$ 
bijectively maps ${\cal N}_{\overline{G}}(u,v)$ to the set of spanning trees of $\overline{G}^\dagger$ that contain $\{s',t'\}$ but not $\{u,v\}^\dagger=\{u',v'\}$, 
and moreover, $(s',t')$ is on the unique $u'v'$-path in the tree. This is a 
bijection by our earlier observation. Finally, define $\phi'(T^\dagger)$ to 
be $T^\dagger$ but with $\{s',t'\}$ removed, and $\{u',v'\}$ added. Then we claim 
that $\phi'\circ\phi$ is a bijection from ${\cal N}_{\overline{G}}(u,v)$ 
to ${\cal N}_{\overline{G}^\dagger}(u',v')={\cal N}_{\overline{G}^\dagger\setminus\{s',t'\}}(u',v')$. 
\end{proof}

\end{document}